\documentclass[10pt,a4paper,reqno]{amsart}
\usepackage{amsthm,amssymb,amstext,graphicx,comment}
\graphicspath{{numerics/results/}}
\newcommand{\floorCF}{1.2248\times10^{-5}}
\newcommand{\floorMC}{1.2286\times10^{-5}}
\newcommand{\floorSE}{4.3511\times10^{-8}}
\newcommand{\floorRelErr}{0.31\%}

\usepackage[foot]{amsaddr}
\usepackage[a4paper, margin=2.5cm]{geometry}
\usepackage{bbm}
\usepackage{enumerate}
\usepackage{amscd}
\usepackage{mathtools}
\usepackage{hyperref}
\hypersetup{
  colorlinks=true,
  pdfpagelayout=TwoColumnRight,
  pdftitle={Hedging Maturity-Specific Risk in Forward Curve Derivatives under Stochastic Volatility},
  pdfauthor={Riccardo Alberti and Sven Karbach},
  pdfsubject={Variance-optimal hedging for forward-curve derivatives},
  pdfkeywords={quadratic hedging, variance-optimal hedging, forward curves, stochastic volatility, HJMM, finite-rank approximation}
}
\usepackage{enumitem}
\usepackage{empheq}
\usepackage[utf8]{inputenc}
\usepackage[T1]{fontenc}
\usepackage{lmodern}
\usepackage{tikz}
\usepackage{eso-pic}
\usepackage{ifthen}
\usepackage{bibentry}
\usepackage{soul}
\usepackage[capitalise]{cleveref}
\usepackage{natbib}
\setcitestyle{numbers,comma,open={[},close={]}}

\newtheorem{theorem}{Theorem}[section]
\newtheorem{corollary}[theorem]{Corollary}
\newtheorem{lemma}[theorem]{Lemma}
\newtheorem{definition}[theorem]{Definition}
\newtheorem{remark}[theorem]{Remark}
\newtheorem{proposition}[theorem]{Proposition}

\newtheoremstyle{example}
  {.3\baselineskip}
  {.3\baselineskip}
  {\normalsize}  
  {0pt}       
  {\bfseries} 
  {.}         
  {5pt plus 1pt minus 1pt} 
  {}          
\theoremstyle{example}
\newtheorem{example}[theorem]{Example}

\newtheorem*{assumption*}{\assumptionnumber}
\providecommand{\assumptionnumber}{}
\makeatletter

\newcommand{\df}{\coloneqq}

\newcommand{\interior}[1]{({\kern0pt#1})^{\textnormal{o}}}
\newcommand{\set}[1]{\left\{ #1\right\}}

\begin{document}

\title[Hedging Maturity-Specific Risk in Forward Curve Derivatives]{Hedging Maturity-Specific Risk in Forward Curve Derivatives under Stochastic Volatility}

\address{\today{}, Korteweg-de Vries Institute for Mathematics and Informatics Institute, University
of Amsterdam, The Netherlands.}     

\maketitle

\vspace{-5mm}

\begin{center}
  \begin{tabular}{cc}
    \textsc{Riccardo Alberti} & \textsc{Sven Karbach} \\
    \small[\texttt{\MakeLowercase{r-alberti@live.it}}] &  \small[\texttt{\MakeLowercase{sven@karbach.org}}] \\
  \end{tabular}
\end{center}

\begin{abstract}
    We study the variance-optimal hedging of European contingent claims written on forwards.
    We assume that
the dynamics of the underlying forward curves follow a Heath--Jarrow--Morton--Musiela
stochastic partial differential equation modulated by an infinite-rank stochastic
covariance component. The variance-optimal hedge is then given by the Galtchouk--Kunita--Watanabe
projection with respect to some covariance-norm quotient generated by the forward curve
martingale. We show density of finite-maturity and
delivery-window strategies, convergence of spectral finite-rank hedge projections and an exact decomposition of the quadratic hedging error into bucket, rank and residual risk components. In enlarged filtrations, the
residual risk is a stochastic-volatility floor for claims loading on non-traded covariance noise. We
illustrate the hedging framework in affine stochastic covariance and multiplicative HJMM
models, and give a concrete example of the
decomposition in a CIR stochastic covariance model.\par{}\medskip{}

\end{abstract}

\maketitle

\section{Introduction}\label{sec:Introduction}
Derivatives whose underlying is a forward rate or forward price arise naturally in both fixed–income and
commodity markets. Under some pricing measure, collateralised
forward or futures prices are martingales, and European payoffs may depend either on one maturity
or on a delivery window. For example, a \emph{caplet} with strike $\kappa$ and payment date $T$ is
a single-maturity option written on a forward rate with payoff 
$$
\big(F_T(\tau^\ast) - \kappa\big)^{+},
\qquad \tau^\ast\geq T\ \text{fixed}.
$$ 
By contrast, a typical \emph{power option} is a European option written on an electricity base-load future
delivering over $[\tau_1,\tau_2]$ that pays
$$
\left(\frac{1}{\tau_2-\tau_1}\int_{\tau_1}^{\tau_2} F_T(u)\,\mathrm{d}u - \kappa\right)^{+},
$$
so the payoff depends on a whole segment of the curve. Both cases fit the same
infinite-dimensional Heath--Jarrow--Morton framework, in which the field $(t,\tau)\mapsto F_t(\tau)$ is
treated as a single forward-curve state variable. More precisely, we denote by $x\df\tau-t$ the time-to-maturity of the forward, and set $f_t(x)\df
F_t(t+x)$. In these coordinates, the risk-neutral evolution of the forward curve is described by the 
Heath--Jarrow--Morton--Musiela (HJMM) equation~\cite[§2.4]{CT06}:
\begin{align}\label{eq:HJMM}
  \begin{cases}
    \mathrm{d} f_t(x)
    = \partial_x f_t(x)\,\mathrm{d}t
      + \sum_{i=1}^{d}\sigma_t^{(i)}(x)\,\mathrm{d}W_t^{(i)}, 
      & t>0,\ x>0, \\
    f_0(x) = F_0(x), & x>0.
  \end{cases}
\end{align}
Here $\{W^{(i)}\}_{i=1,\ldots,d}$ denotes a family of independent Brownian motions with 
$d \in \mathbb{N} \cup \{+\infty\}$. The parameter $d$ is the rank of the noise: finite $d$
gives a finite-rank HJM model, while $d=\infty$ gives an infinite-rank model. In this framework, European
claims written on forwards are interpreted as nonlinear functionals of the state variable $x\mapsto
f_T(x)$. In this paper, we develop a quadratic hedging theory for such claims in HJMM models with both stochastic
volatility and infinite-rank noise, i.e., the family $\set{\sigma^{(i)}\colon i=1,\ldots,d}$ may
vary randomly in time and we let $d=\infty$. Note that empirical evidence from energy and commodity forward markets supports high-dimensional,
maturity-dependent risk structures with time-varying correlations~\cite{AKW10}. Principal
component analyses of NordPool power futures require more than ten factors to explain $95\%$ of
total variance~\cite{Koekebakker2005ForwardCD,Fre08,BBK08}, in contrast with many interest-rate
term structures where a few factors explain most variation~\cite[§1.7.3]{CT06}.
Explained variance, however, is not the appropriate criterion for hedging. A direction with low
unconditional variance may still materially affect option values, hedge ratios, quantiles, or
value-at-risk whenever the payoff is sufficiently exposed to that
direction~\cite{Con05}. Consequently, an a priori rank reduction based solely on principal-component
analysis may lead to substantial hedging errors. The correct and relevant notion of approximation
error is instead determined by the covariance-based hedging norm introduced in this paper. More
precisely, we shall formulate the hedging problem in a full infinite-rank setting and
assess any subsequent rank reduction directly in the metric relevant for hedging.

\subsection{Stochastic Covariance and Maturity-Specific Risk}

Throughout the paper, we work on a filtered probability space supporting a cylindrical Brownian motion $W$ on a separable Hilbert space $G$. The instantaneous volatility is represented by a predictable process
\begin{align*}
  \sigma_t \in \mathcal{L}_2(G,H),
\end{align*}
where $H$ denotes the forward-curve state space. Accordingly, the noise term in~\eqref{eq:HJMM} is understood as the Hilbert-space stochastic integral $\sigma_t\,\mathrm{d}W_t$. The corresponding instantaneous covariance operator is
\begin{align*}
\Sigma_t \coloneqq \sigma_t \sigma_t^\ast \in \mathcal{L}(H),
\qquad t \geq 0.
\end{align*}
For every $t \geq 0$, the operator $\Sigma_t$ is self-adjoint, non-negative, and trace class, but it may have infinite rank. We refer to the joint process
$(f_t,\Sigma_t)_{t\geq 0}$ as a \emph{stochastic-volatility-modulated HJMM model}.

Infinite-rank covariance provides a natural representation of maturity-specific risk. In an
idealised market with a continuum of maturity dates, distinct regions of the forward curve may be
exposed to distinct sources of risk. Therefore, capturing such effects generally requires infinitely
many factors ($d=\infty$), see also the discussion~\cite[Section~6.5.3]{CT06}. 

\subsection{Hedging European Contingent Claims in Forward Markets}

Let $K = h\bigl(F_T(\cdot)\bigr)$ be a European contingent 
claim in $L^2(\Omega,\mathcal{F}_T,\mathbb{P})$ written on the forward curve at time $T$. This class
includes both point-maturity payoffs and payoffs depending on averages over delivery windows. We
assume that trading is possible in a strip of forward or futures contracts, with admissible
self-financing gains defined in Section~\ref{sec:trading-strategies} below. Among all admissible
strategies and initial endowments, we minimise the mean-square hedging error. We refer to the resulting optimisation problem as the \emph{variance-optimal hedging problem}. Let
$$
M_t \coloneqq \int_0^t \sigma_s\,\mathrm{d}W_s,
\qquad t\geq 0,
$$
denote the $H$-valued martingale component of the forward-curve dynamics in~\eqref{eq:HJMM}. Since the forward curve is a martingale up to the Musiela shift under the pricing measure, the variance-optimal hedge is characterised by the Galtchouk--Kunita--Watanabe (GKW) decomposition (Theorem~\ref{thm:variance-optimal-hedging}):
$$
K
=
\mathbb{E}[K]
+
\int_0^T \phi_s^\star\,\mathrm{d}M_s
+
N_T,
\qquad
N\perp M.
$$
Here $\phi^\star$ is the predictable GKW integrand and $N$ a square-integrable martingale strongly orthogonal to $M$; the initial capital $\mathbb{E}[K]$ together with $\phi^\star$ minimises the mean-square hedging error, with minimal residual risk $\mathbb{E}\bigl[N_T^2\bigr]$.

Maturity-specific risk has two conceptually distinct consequences for this problem. \emph{First}, even when the infinite-rank volatility is deterministic, exact replication by strictly viable portfolios involving only finitely many maturities may fail. Writing $Q_M(t)$ for the martingale covariance density introduced below, its range is dense in its closed support,
$\overline{\operatorname{Ran} Q_M(t)}=\bigl(\ker Q_M(t)\bigr)^\perp$,
but need not itself be closed. The pointwise optimal hedge may then involve an unbounded
pseudo-inverse and exist only as a limit of admissible hedges in the covariance norm, so the
abstract GKW integrand $\phi^\star$ need not be implementable by finitely many traded maturities or
delivery periods. This is a failure of strict viability, not of market completeness: under
deterministic volatility, a Brownian filtration, and the usual covariance-support condition the market
remains \emph{approximately complete} (Corollary~\ref{cor:approx-completeness};
Example~\ref{ex:infinite-rank-nonviable} works out an explicit hedge of this kind).

\emph{Second}, a strictly positive irreducible hedging error arises only when the covariance is
stochastic and its additional source of randomness affects the claim without being spanned by
trading in the forward curve. This unspanned covariance risk generates the
\emph{stochastic-volatility floor} $\mathbb{E}[N_T^2]$ of the quadratic hedging problem. Along a
nested family of closed approximation spaces, the mean-square hedging error then admits an exact
decomposition into
\begin{enumerate}[label=(\roman*),leftmargin=*]
    \item a finite-bucket implementation error;
    \item a finite-rank spectral truncation error; and
    \item the irreducible residual risk $\mathbb{E}[N_T^2]$,
\end{enumerate}
separating the cost of the traded contracts from that of covariance-rank reduction and from the
intrinsic incompleteness generated by unspanned risk (Proposition~\ref{thm:three-way}). The floor is
made fully explicit for an affine covariance model driven by a Cox--Ingersoll--Ross process in
Section~\ref{subsubsec:closed-form-floor}.

\subsection{Contributions and Related Literature}

This paper develops a quadratic hedging theory for derivatives written on forward curves in
infinite-dimensional HJMM models with operator-valued stochastic covariance. We work throughout
under a fixed pricing measure under which the traded forward or futures contracts are
martingales. Thus, we consider the \emph{martingale case} of quadratic hedging, in which the
variance-optimal hedging problem reduces to an orthogonal projection in $L^2$, and the
variance-optimal, risk-minimising, and Galtchouk--Kunita--Watanabe strategies coincide under the
standing integrability assumptions. Our objective is not to extend the general semimartingale theory
of variance-optimal hedging or variance-optimal martingale measures developed, for example,
in~\cite{Schweizer_2001,CernyKallsen2007}. Nor is the existence of a GKW decomposition for
Hilbert-space-valued martingales itself new; infinite-dimensional martingale representation and projection results go back at least to~\cite{Ouvrard1975ReprsentationDM}. The main issue addressed here is instead the stability and implementability of the GKW projection when the forward-curve covariance is compact, stochastic, operator-valued, and potentially of infinite rank. In this setting, the relevant covariance operators generally have non-closed range, their pseudo-inverses may be unbounded, and the formal optimal integrand need not correspond directly to a portfolio of finitely many traded contracts. More specifically, the paper makes the following contributions:
\begin{enumerate}[label=(\roman*),leftmargin=*]
    \item We introduce a covariance-norm quotient space for forward-curve gains (Definition~\ref{def:Lambda2-star}), thereby removing covariance-null directions before invoking any pseudo-inverse; this is a covariance-weighted analogue of the generalised integrands of~\cite{BjorkDiMasiKabanovRunggaldier1997,DeDonnoGuasoniPratelli2005}.
    
    \item We prove that trading strategies based on finitely many point maturities and delivery windows are dense in this quotient space (Proposition~\ref{prop:density} and Lemma~\ref{lem:predictable-finite-maturity-density}).
    
    \item We establish convergence of spectral approximations to the GKW integrand under predictable
commuting projections (Theorem~\ref{thm:finite-rank-stability}) and illustrate it in examples.
    
    \item We identify the GKW residual with an irreducible stochastic-volatility hedging floor when the claim is exposed to covariance noise that cannot be spanned by trading in the forward curve, and separate this floor from the finite-rank degeneracy of the covariance operator (Proposition~\ref{prop:dichotomy}).
\end{enumerate}

The analysis extends the classical infinite-dimensional HJM and SPDE framework of~\cite{Fil01,CT06} to stochastic-covariance environments in which the volatility takes values in the Hilbert--Schmidt operators. The theory covers both claims depending on a single maturity, such as caplet-type payoffs, and claims depending on delivery windows, such as options on electricity base-load futures. We illustrate the assumptions using two model classes: affine models with operator-valued stochastic covariance and multiplicative HJMM models with fixed covariance eigenfunctions and state-dependent eigenvalues.

Our approach builds on several strands of the literature. The state-space formulation, Musiela
parametrisation, and curve-valued SPDE representation follow the Hilbert-space framework of
Filipovi\'c~\cite{Fil01} and the infinite-dimensional HJM theory of Carmona and
Tehranchi~\cite{CT06}. In commodity markets, Benth and Kr\"uhner~\cite{BenthKruehner2014CMS}
represent forward-curve dynamics through covariance operators given by integral kernels on
time-to-maturity spaces, providing a natural foundation for the operator-valued, potentially
infinite-rank stochastic covariance models studied here and matching the high-dimensional,
maturity-dependent covariance structures documented in electricity forward
markets~\cite{Koekebakker2005ForwardCD,BorakWeron2008,Benth2008}.

Variance-optimal and quadratic hedging in general semimartingale models is well established; see,
among others, Schweizer~\cite{Schweizer_2001}, Pham~\cite{Pham2000}, and \v{C}ern\'y and
Kallsen~\cite{CernyKallsen2007}. Explicit solutions for affine models were studied by Kallsen and
Pauwels~\cite{Kallsen_Pauwels_2010} and by the authors in a multivariate stochastic covariance setting in~\cite{CK26}. Relative to this literature, we deliberately restrict attention to the martingale case under the pricing measure. This allows us to isolate the infinite-dimensional geometric and approximation problems created by compact covariance operators. The novelty relative to the abstract GKW projection lies in the covariance-norm quotient construction, its treatment of non-closed covariance ranges and unbounded pseudo-inverses, and the accompanying finite-rank and finite-bucket approximation theory.

The idea of admitting a strategy that is not a finite tradable portfolio but an element of a
completion is itself rooted in the theory of bond markets with a continuum of maturities. Generalised
portfolios across a maturity continuum were introduced by Bj\"ork, Di Masi, Kabanov and
Runggaldier~\cite{BjorkDiMasiKabanovRunggaldier1997}; a corresponding stochastic-integration theory
for bond markets, in which the admissible integrands form a closure that need not consist of pointwise
portfolios, was developed by De Donno and Pratelli and by De Donno, Guasoni and
Pratelli~\cite{DeDonnoPratelli2005,DeDonnoGuasoniPratelli2005}, while attainability and completeness
of such Hilbert-space bond markets were analysed by Taflin and by Ekeland and
Taflin~\cite{Taflin2005,EkelandTaflin2005}. Relative to this line of work, our completion is weighted
by the covariance norm~\eqref{eq:lambda-norm}, is built for operator-valued, potentially
infinite-rank \emph{stochastic} covariance, and is used to solve the \emph{quadratic-hedging}
projection rather than questions of existence, super-replication, or exact completeness. This is what
produces the bucket/rank/residual decomposition of Proposition~\ref{thm:three-way} and the explicit
stochastic-volatility floor, neither of which is visible at the level of the abstract generalised
integrand.

In terms of infinite-dimensional modelling, the paper is also related to the literature on infinite-dimensional stochastic-volatility
models for forward curves. Relevant developments include affine stochastic-volatility models on
Hilbert spaces~\cite{BS18,BenthRudigerSuess2018, CoxKarbachKhedher2022, FK24, HKK25}, locally
state-dependent HJMM models~\cite{DL25}, robustness results for Hilbert-space
stochastic-volatility processes~\cite{Benth2024Robustness, Kar26}, and heat-modulated affine
stochastic-volatility models for forward-curve dynamics~\cite{Karbach2024HeatAffine}. To our
knowledge, the present paper is the first to connect these operator-valued stochastic-covariance
modelling developments with variance-optimal hedging.

\subsection{Layout of the Article}

Section~\ref{sec:stoch-volat-modul-1} introduces the forward-curve state space, states the standing
assumptions, and defines both generalised and realistic admissible trading strategies together with their self-financing gains. Section~\ref{sec:Variance-Optimal-Hedging} develops the GKW decomposition and characterises the variance-optimal integrand in the covariance-norm quotient space. Section~\ref{sec:appr-real-portf} establishes convergence of the spectral and finite-bucket approximations and derives the three-way decomposition of the hedging error. Section~\ref{sec:bns-specialization} presents affine stochastic-covariance models and multiplicative fixed-eigenbasis HJMM models covered by the general theory. Section~\ref{sec:numerics} describes the finite-factor numerical implementation.

\section{The Stochastic Volatility Modulated HJMM Financial Market}
\label{sec:stoch-volat-modul-1}

This section fixes the forward-curve state space, the HJMM dynamics, and the self-financing gain
space used throughout the paper. The material in Section~\ref{sec:forward-curve-dynamics} and the
mild well-posedness of Proposition~\ref{prop:mild-solution} are largely standard and are included to
keep the paper self-contained. Section~\ref{sec:trading-strategies} then contains new results on
trading strategies and the norms appropriate for evaluating them.

\subsection{The State Space}
\label{sec:forward-curve-dynamics}

Fix a maturity range $[0,\Theta_{\max}]$ with $\Theta_{\max}\in(0,\infty]$. 
Empirically, for each trading time $t$, the map $x\mapsto f_t(x)$ is sufficiently smooth, admits a finite long rate, 
and \emph{flattens} as $x\uparrow\Theta_{\max}$; see, e.g., \cite{Fil01,Con05,Dou14}. 
To encode these features, let $w\colon[0,\Theta_{\max}]\to(0,\infty)$ be continuously differentiable, strictly positive, 
non--decreasing, with $w^{-1}\in L^1_{\mathrm{loc}}([0,\Theta_{\max}))$. 
When $\Theta_{\max}=\infty$, assume also the long-rate condition
$w^{-1}\in L^1([0,\infty))$ and
$$
\sup_{x\geq 0} w(x)^{-1}\Big(1+\int_0^x w(y)^{-1}\,\mathrm{d}y\Big)^2<\infty,
$$
cf.\ \cite[Ch.~4]{Fil01}. 
Define
$$
H_{w,\Theta_{\max}}
:=
\Big\{ f\in \mathrm{AC}_{\mathrm{loc}}\big([0,\Theta_{\max})\big)\colon
\|f\|_{w,\Theta_{\max}}^2
:= |f(0)|^2+\int_0^{\Theta_{\max}} w(x)\,|f'(x)|^2\,\mathrm{d}x
<\infty \Big\},
$$
with inner product
$$
\langle f,g\rangle_{w,\Theta_{\max}}
:= f(0)g(0)+\int_0^{\Theta_{\max}} w(x)\,f'(x)g'(x)\,\mathrm{d}x.
$$
Then $(H_{w,\Theta_{\max}},\langle\cdot,\cdot\rangle_{w,\Theta_{\max}})$ is a separable Hilbert space. 
For each $x\in[0,\Theta_{\max}]$, the evaluation map $f\mapsto f(x)$ is continuous; hence there exists a unique 
$u_x\in H_{w,\Theta_{\max}}$ with $\langle f,u_x\rangle_{w,\Theta_{\max}}=f(x)$, for all $f\in  H_{w,\Theta_{\max}}$. 
If $\Theta_{\max}=\infty$, the assumptions imply the existence of $f(\infty):=\lim_{x\to\infty} f(x)$ 
and the flattening $|f(x)-f(\infty)|\to 0$ as $x\to\infty$; 
when $\Theta_{\max}<\infty$, $f(\Theta_{\max})$ is well defined and evaluation at $\Theta_{\max}$ is continuous.

We develop the theory on the unbounded horizon $\Theta_{\max}=\infty$, the standard Filipovi\'c
space, and write $H_w:=H_{w,\infty}$; the finite-horizon case is recorded as
Remark~\ref{rmk:finite-horizon} below. On $H_w$ the Musiela shift semigroup
$$
(S(t)f)(x):=f(x+t),\qquad x,t\ge 0,
$$
is strongly continuous under the standing weight assumptions \cite[Thm.~5.1.1]{Fil01}. Its
infinitesimal generator is $\mathcal{A}=\partial_x$ with natural domain
$$
\mathcal{D}(\mathcal{A})
=
\Big\{ f\in H_{w,\Theta_{\max}}\colon\ f'\in H_{w,\Theta_{\max}} \Big\}.
$$
The decomposition $f \mapsto (f(0), f')$ identifies $H_{w,\Theta_{\max}}$ isometrically with
$\mathbb{R} \oplus L^2\bigl([0,\Theta_{\max}], w(x)\,\mathrm{d}x\bigr)$. Compactness is not needed
for the abstract hedging construction. When compact embeddings are useful for finite-dimensional
approximation, they must be taken into a weaker topology. For example, if $\Theta_{\max}=\infty$,
exponential damping yields: for any
$\gamma>0$, the embedding
$$
H_{w,\infty}
\hookrightarrow
\mathbb{R} \oplus L^2\bigl([0,\infty),\,w(x)\mathrm{e}^{-\gamma x}\,\mathrm{d}x\bigr)
$$
is compact, see~\cite[Thm.~2.3]{Tap13} and~\cite[Section 2]{Karbach2024HeatAffine}.

\begin{remark}[Finite maturity horizon]\label{rmk:finite-horizon}
For a finite horizon $\Theta_{\max}<\infty$ the right shift maps out of
$[0,\Theta_{\max}]$, so a strongly continuous Musiela semigroup requires a boundary
condition. A standard choice is the absorbing Musiela shift on $H_{w,\Theta_{\max}}$ with the
extension $f\equiv f(\Theta_{\max})$ for $x\ge\Theta_{\max}$, which is a $C_0$-semigroup
with generator $\partial_x$ on $\{f\in H_{w,\Theta_{\max}}:f'\in H_{w,\Theta_{\max}},\ f'(\Theta_{\max})=0\}$;
see~\cite[Section 2.3]{farkas2015isem}. All hedging results below use only the existence of a
strongly continuous Musiela semigroup and the martingale part $M$, so they transfer verbatim to any
such finite-horizon realisation. The dual functionals $\delta_x$ and $\ell_{[a,b]}$ are unchanged,
so the density results for point maturities and delivery windows transfer as well. The illustrative
example of Section~\ref{subsec:synthetic-study} uses a finite horizon and depends only on the
covariance operator, not on the choice of boundary condition.
\end{remark}

\subsection{HJMM Dynamics and Mild Well-Posedness}

Let $(\Omega,\mathcal{F},\mathbb{F},\mathbb{P})$ be a filtered probability space carrying a 
cylindrical Brownian motion $W$ on a separable Hilbert space $G$. 
The forward curve (under Musiela parametrisation) takes values in $H_w$ and satisfies the abstract HJMM SPDE
\begin{align}\label{eq:HJMM-SPDE}
\mathrm{d}f_t
=
\bigl(\mathcal{A} f_t + g_t \bigr)\,\mathrm{d}t
+
\sigma_t\,\mathrm{d}W_t,
\qquad
f_0\in H_w,
\end{align}
where $\mathcal{A}=\partial_x$ is the generator of $\{S(t)\}_{t\ge 0}$ on $H_w$, 
$g=\{g_t\}_{t\ge 0}$ is an $H_w$-valued progressively measurable drift, 
and $\sigma=\{\sigma_t\}_{t\ge 0}$ is predictable with $\sigma_t\in \mathcal{L}_2(G,H_w)$ 
(Hilbert--Schmidt operators from $G$ to $H_w$). 
Since $W$ is cylindrical on $G$ and $\sigma_t$ is Hilbert--Schmidt, the noise enters via the
stochastic convolution. 

In the arithmetic (collateralised-futures) case, $g\equiv 0$ so discounted futures are martingales. 
In the geometric (forward-price) case, $g$ is fixed by the HJM no--arbitrage condition so that suitably discounted forwards 
are martingales for all maturities.

Fix $T>0$ and assume:
\begin{itemize}
\item[(A1)] $f_0\in H_w$ is $\mathcal{F}_0$-measurable with $\mathbb{E}\big[\|f_0\|_{H_w}^2\big]<\infty$;
\item[(A2)] $g$ is predictable with $\displaystyle \mathbb{E}\big[\int_0^T \|g_t\|_{H_w}^2\,\mathrm{d}t\big]<\infty$;
\item[(A3)] $\sigma$ is predictable with
$\displaystyle \mathbb{E}\big[\int_0^T \|\sigma_t\|_{\mathcal{L}_2(G,H_w)}^2\,\mathrm{d}t\big]<\infty$.
\end{itemize}

The formulation covers both exogenous stochastic volatility and multiplicative HJMM noise. Indeed,
one may first construct a mild solution of
$$
\mathrm{d}f_t=(\mathcal A f_t+g_t)\,\mathrm{d}t+\Gamma(t,f_t,Y_t)\,\mathrm{d}W_t
$$
for a predictable state process $Y$ and a coefficient
$\Gamma:[0,T]\times H_w\times E\to\mathcal L_2(G,H_w)$, and then set
$\sigma_t:=\Gamma(t,f_t,Y_t)$. Standard Lipschitz and linear-growth assumptions on $\Gamma$
ensure (A3) and mild well-posedness; for instance, it suffices that
$$
\|\Gamma(t,f,y)-\Gamma(t,\tilde f,y)\|_{\mathcal L_2(G,H_w)}
\le L\|f-\tilde f\|_{H_w},
\qquad
\|\Gamma(t,f,y)\|_{\mathcal L_2(G,H_w)}
\le C(1+\|f\|_{H_w}+\|y\|_E),
$$
together with the corresponding square-integrability of $Y$. Thus the hedging theory below depends
on the realised martingale part $M=\int\sigma\,\mathrm{d}W$, not on whether the volatility is
exogenous or state-dependent. The following proposition is a special case of \cite[Thm.~5.4]{DaPratoZabczyk2014}.

\begin{proposition}[Mild well-posedness]\label{prop:mild-solution}
Under \emph{(A1)--(A3)}, there exists a unique $H_w$-valued, $\mathbb{F}$-adapted process
with continuous paths solving~\eqref{eq:HJMM-SPDE} in the mild sense:
\begin{align}\label{eq:mild-solution}
f_t
=
S(t)f_0
+\int_0^t S(t-s)\,g_s\,\mathrm{d}s
+\int_0^t S(t-s)\,\sigma_s\,\mathrm{d}W_s,
\qquad t\in[0,T].
\end{align}
Moreover, the Musiela shift is a $C_0$-semigroup, so
$M_T^S:=\sup_{0\le t\le T}\|S(t)\|_{\mathcal{L}(H_w)}<\infty$ and there exists a constant $C_T<\infty$ depending only on $T$ and $M_T^S$ with
\begin{align}\label{eq:apriori}
\sup_{t\in[0,T]}\mathbb{E}\big[\|f_t\|_{H_w}^2\big]
\le
C_T\left(
\mathbb{E}\big[\|f_0\|_{H_w}^2\big]
+\mathbb{E}\big[\int_0^T \|g_s\|_{H_w}^2\,\mathrm{d}s\big]
+\mathbb{E}\big[\int_0^T \|\sigma_s\|_{\mathcal{L}_2(G,H_w)}^2\,\mathrm{d}s\big]
\right).
\end{align}
In particular, $M_t = \int_0^t \sigma_s\,\mathrm{d}W_s$ is a square-integrable $H_w$-valued
martingale with continuous paths.
\end{proposition}

Let $(e_i)_{i\ge 1}$ be an orthonormal basis of $G$ with independent standard Brownian motions $(\beta^i)_{i\ge 1}$ 
so that formally $W_t=\sum_{i\ge 1} e_i\,\beta_t^i$. Then
$$
\sigma_t\,\mathrm{d}W_t
=
\sum_{i\ge 1} \sigma_t e_i\,\mathrm{d}\beta_t^i,
\qquad
\sigma_t^{(i)}:=\sigma_t e_i\in H_w.
$$
and we see that finite rank of $\sigma_t$ yields an $n$-factor HJM representation. Other operator-theoretic preliminaries (the bounded, compact, Hilbert--Schmidt and trace-class ideals
$\mathcal{L}(G,H)$, $\mathcal{K}(G,H)$, $\mathcal{L}_2(G,H)$, $\mathcal{L}_1(G,H)$, together with
the relevant composition bounds) are collected in Appendix~\ref{app:operator-prelim}.

\subsection{Dynamic Trading in Forward Markets}
\label{sec:trading-strategies}

In the idealised forward market the investable universe is continuously indexed by maturity: at
time $t\in[0,T]$, one may trade the family $f_t(x)=F(t,t+x)$ for $x\ge0$. We model
forward-curve positions as continuous linear functionals on the state space $H$.

Let $H^\ast$ denote the dual space of $H$, and $\langle\cdot,\cdot\rangle_{H^\ast,H}$ the duality pairing. 
Since point evaluations are continuous on $H$, we have $\delta_x \in H^\ast$ for all $x \ge 0$. 
Choosing $\phi_t \in H^\ast$ as the forward-curve position ensures that the portfolio wealth depends continuously on the forward prices $f_t \in H$.

\begin{definition}\label{def:generalized-trading-strategy-wealth}
A \emph{trading strategy} is a predictable process $\Phi_t = (\psi_t, \phi_t) \in \mathbb{R} \times H^\ast$, 
where $\psi_t$ denotes the bank-account position and $\phi_t$ the forward-curve position. 
The associated \emph{wealth process} is defined by
$$
U_t(\Phi)
=
\psi_t
+
\langle \phi_t, f_t\rangle_{H^\ast,H},
\qquad 0 \leq t \leq T,
$$
where $f_t \in H$ denotes the Musiela forward curve.
\end{definition}

\begin{example}[Single–maturity forward]\label{ex:single-forward}
Holding one forward contract with delivery at date $T$ (corresponding to time-to-maturity $T-t$) together with the appropriate cash compensator yields
$$
U_t 
= 
\delta_{\,T-t}(f_t) - F(0,T),
\qquad 0 \leq t \leq T,
$$
which is consistent with Definition~\ref{def:generalized-trading-strategy-wealth}, 
since $\delta_{\,T-t} \in H^\ast$.
\end{example}

\paragraph{Point forwards and delivery-window contracts.}
The single-maturity functional $\delta_x$ models an idealised forward with a single delivery date
$t+x$, which is the natural traded instrument in interest-rate and futures markets. In commodity
and especially electricity markets, however, the liquidly traded contracts deliver over a
\emph{period}: a base-load contract over the time-to-delivery window $[a,b]$ pays the average price
$\frac{1}{b-a}\int_a^b F_t(t+u)\,\mathrm du$ (see~\cite[Ch.~4]{BBK08} and~\cite{Koekebakker2005ForwardCD}
for the structure and dynamics of such delivery-period swaps), so its forward-curve position is the
\emph{delivery-window averaging functional}
\begin{align}\label{eq:window-functional}
\ell_{[a,b]}:=\frac{1}{b-a}\int_a^b\delta_u\,\mathrm du\in H^\ast,
\qquad
\ell_{[a,b]}(f)=\frac{1}{b-a}\int_a^b f(u)\,\mathrm du,
\quad 0\le a<b\le\Theta_{\max}.
\end{align}
Boundedness $\ell_{[a,b]}\in H^\ast$ is immediate: since evaluation is continuous, the representer
$u_{[a,b]}\in H$, defined by the Bochner integral $u_{[a,b]}:=\frac1{b-a}\int_a^b u_x\,\mathrm dx$,
satisfies $\langle f,u_{[a,b]}\rangle_H=\ell_{[a,b]}(f)$, so that
$$
\|\ell_{[a,b]}\|_{H^\ast}=\|u_{[a,b]}\|_H\le\sup_{x\in[a,b]}\|u_x\|_H<\infty .
$$
Point forwards are recovered as the degenerate limit $\ell_{[a,b]}\to\delta_a$ as $b\downarrow a$,
since $u_{[a,b]}\to u_a$ in $H$ by continuity of $x\mapsto u_x$.

In practice, investors can hold only finitely many such contracts at each time $t$, which motivates
the notion of a \emph{realistic trading strategy}. We allow both instrument types: point forwards
(interest-rate/futures markets) and delivery-window averages (energy markets).

\begin{definition}\label{def:realistic-trading-strategy}
We call a strategy $\Phi_t=(\psi_t,\phi_t)$ \emph{realistic} if
$$
\phi_t = \sum_{i=1}^{n(t)} q_i(t)\,\ell_i,
\qquad
\ell_i\in\{\delta_{\,x_i(t)}\}\cup\{\ell_{[a_i,b_i]}\},
$$
a finite predictable combination of point-forward functionals and/or delivery-window averaging
functionals~\eqref{eq:window-functional}, with $n(t)\in\mathbb{N}$ and predictable weights
$q_i(t)\in\mathbb{R}$. For the density arguments below it is enough to allow deterministic maturity
grids (resp.\ fixed delivery windows) and predictable weights.
\end{definition}

In the idealised continuum market, $\delta_x$ denotes the forward or futures contract with current
time-to-maturity $x$. A simple strategy that holds $\delta_x$ over a time interval is therefore read
as a frictionless rolling position in the corresponding maturity bucket. Fixed calendar-date
buy-and-hold forwards, with functional $\delta_{\tau-t}$, are included as elementary examples, but
they do not exhaust the idealised gain space used for the density result. A fixed exchange ladder of
delivery periods is treated separately through the bucket gap in
Remark~\ref{rmk:delivery-buckets}.

\begin{center}
\begin{tabular}{ll}
\hline
Term & Meaning\\
\hline
Generalised strategy & predictable integrand in the covariance-norm completion\\
Realistic strategy & finite predictable combination of traded point or window functionals\\
Finite-maturity strategy & realistic strategy using finitely many point forwards\\
Finite-bucket strategy & realistic strategy using a fixed delivery ladder\\
Strictly viable strategy & pointwise $H^\ast$-valued representative\\
Implemented strategy & finite-rank hedge projected onto available market instruments\\
\hline
\end{tabular}
\end{center}

Variance–optimal hedges derived later live in the covariance-norm completion and need not be
realistic pointwise. Finite-maturity rolling portfolios are nevertheless dense in the relevant
hedging norm, and the same density holds for refinable delivery-window systems. A fixed exchange
ladder is finite and need not be dense; its residual bucket gap is treated separately in
Remark~\ref{rmk:delivery-buckets}.

\begin{proposition}[Static density of maturity evaluations and window averages]\label{prop:density}
The linear span of $\{\delta_x\}_{x\geq 0}$ is dense in $H^\ast$, and so is the linear span of the
delivery-window averaging functionals $\{\ell_{[a,b]}:0\le a<b\le\Theta_{\max}\}$.
\end{proposition}

\begin{proof}
Let $J\colon H\to H^\ast$ be the Riesz isometry and $u_x:=J^{-1}\delta_x\in H$ the representer of
evaluation at $x$; by the reproducing property, $\langle f,u_x\rangle_H=f(x)$ for all $f\in
H$. Since $J$ is an isometry, density of $\mathrm{span}\{\delta_x\}$ in $H^\ast$ is equivalent to
density of $\mathrm{span}\{u_x\}$ in $H$, which holds in any reproducing kernel Hilbert space as
kernel sections span a dense subspace. Hence $\overline{\mathrm{span}\{\delta_x\}}=H^\ast$.
For the window functionals, $\ell_{[a,a+\varepsilon]}\to\delta_a$ in $H^\ast$ as
$\varepsilon\downarrow0$ (shown above), so the closed span of
$\{\ell_{[a,b]}\}$ contains every $\delta_a$, and therefore equals $H^\ast$ as well.
\end{proof}

\begin{lemma}[Predictable finite-maturity density]\label{lem:predictable-finite-maturity-density}
The predictable simple processes with values in finite linear spans of the evaluation
functionals $\{\delta_x:x\ge0\}$ are dense in
$\Lambda_T^{2}(H,\mathbb R;M)$. The same holds with the delivery-window averaging functionals
$\{\ell_{[a,b]}\}$ of~\eqref{eq:window-functional} in place of the evaluations.
\end{lemma}

\begin{proof}
By definition, $\Lambda_T^{2}(H,\mathbb R;M)$ is the covariance-norm completion of simple
predictable $H^\ast$-valued processes. It therefore suffices to approximate an arbitrary simple
process
$$
  \phi_t(\omega)=\sum_{j=1}^J \xi_j(\omega)\ell_j\,\mathbf 1_{(s_j,t_j]}(t),
$$
where $\xi_j$ are bounded $\mathcal F_{s_j}$-measurable random variables and
$\ell_j\in H^\ast$. By Proposition~\ref{prop:density}, choose
$\ell_j^{(m)}\in\mathrm{span}\{\delta_x:x\ge0\}$ (or, for the window statement,
$\ell_j^{(m)}\in\mathrm{span}\{\ell_{[a,b]}\}$) with
$\|\ell_j^{(m)}-\ell_j\|_{H^\ast}\to0$ for each $j$, and set
$$
  \phi_t^{(m)}(\omega)=\sum_{j=1}^J \xi_j(\omega)\ell_j^{(m)}\,\mathbf 1_{(s_j,t_j]}(t).
$$
Then $\phi^{(m)}$ is predictable and finite-maturity valued. Moreover,
$$
\|\phi^{(m)}-\phi\|_{\Lambda_T^2(M)}^2
\le
\mathbb E\int_0^T
\|\phi_t^{(m)}-\phi_t\|_{H^\ast}^2\,
\|Q_M^{1/2}(t)\|_{\mathcal L_2(H)}^2\,\mathrm dt .
$$
Here
$$
\|Q_M^{1/2}(t)\|_{\mathcal L_2(H)}^2
=\mathrm{Tr}\,Q_M(t)=\sum_k\lambda_k(t),
$$
which is not the Hilbert--Schmidt square
$\|Q_M(t)\|_{\mathcal L_2(H)}^2=\sum_k\lambda_k(t)^2$.
The integrand is dominated by a constant multiple of
$\mathrm{Tr}\,Q_M(t)=\|\sigma_t\|_{\mathcal L_2(G,H)}^2$, which is integrable by (A3), and it
converges pointwise to zero. Dominated convergence yields the claim.
\end{proof}

\subsubsection{The Self-Financing Condition and the Gains Process}

Self-financing gains are defined through stochastic integrals against the martingale component of
the forward curve.

Let $H$ denote the forward-curve space from Section~\ref{sec:forward-curve-dynamics}. 
Recall the HJMM dynamics in the abstract SPDE form~\eqref{eq:HJMM-SPDE}, and define the $H$-valued, square-integrable martingale component of the forward curve by
$$
M_t \coloneqq \int_0^t \sigma_s\,\mathrm{d}W_s,
\qquad t \in [0,T].
$$
Then, for all $h, k \in H$, the quadratic covariation satisfies
$$
\big[\langle M, h\rangle_H,\, \langle M, k\rangle_H\big]_t
=
\int_0^t \langle \Sigma_s h,\, k\rangle_H\,\mathrm{d}s,
\qquad
\Sigma_s \coloneqq \sigma_s \sigma_s^{\ast} \in \mathcal{L}(H),
$$
so that $\Sigma_s$ is the predictable, self-adjoint, non-negative covariance operator of $M$ at time $s$.
Equivalently, the (scalar) Doléans measure associated with $M$ is
$\lambda_M(\mathrm{d}s, \mathrm{d}\omega) = \|\sigma_s\|_{\mathcal{L}_2(G,H)}^{2}\,\mathrm{d}s\,\mathbb{P}(\mathrm{d}\omega) = \mathrm{Tr}(\Sigma_s)\,\mathrm{d}s\,\mathbb{P}(\mathrm{d}\omega)$,
while the operator-valued angle-bracket measure is
$\alpha_M(\mathrm{d}s, \mathrm{d}\omega) = \Sigma_s\,\mathrm{d}s\,\mathbb{P}(\mathrm{d}\omega)$ with values in the trace-class ideal.

In this setting, the self-financing condition links the evolution of the wealth process 
$U_t(\Phi)$ from Definition~\ref{def:generalized-trading-strategy-wealth} to the stochastic integral with respect to $M$. 
Trading gains and losses thus arise solely from the dynamics of the forward curve, 
without any external capital injections or withdrawals.

We use the standard Hilbert-space stochastic integral with respect to cylindrical Brownian motion;
see, for example,~\cite{DaPratoZabczyk2014,PeszatZabczyk2007,Fil01}. A bounded forward-curve portfolio $\phi_s\in H^\ast=\mathcal L(H,\mathbb R)$ acts on the martingale
$M$ through the composition $\phi_s\sigma_s$. Equivalently, since
$Q_M(s):=\Sigma_s=\sigma_s\sigma_s^\ast$, its quadratic risk is measured by
$\phi_sQ_M^{1/2}(s)$. The relevant norm is deterministic and comes from the Itô isometry.

\begin{definition}[Admissible integrands and covariance norm]
\label{def:Lambda2-star}
Let $\mathcal E(H^\ast)$ be the vector space of $H^\ast$-valued simple predictable processes on
$[0,T]$. For $\phi,\psi\in\mathcal E(H^\ast)$ set
\begin{align}
\langle\phi,\psi\rangle_{\Lambda_T^2(M)}
&:=
\mathbb E\int_0^T
\langle \phi_t Q_M^{1/2}(t),\psi_t Q_M^{1/2}(t)\rangle_{\mathcal L_2(H,\mathbb R)}
\,\mathrm dt, \label{eq:lambda-inner-product}\\
\|\phi\|_{\Lambda_T^2(M)}^2
&:=
\mathbb E\int_0^T
\|\phi_t Q_M^{1/2}(t)\|_{\mathcal L_2(H,\mathbb R)}^2\,\mathrm dt
=
\mathbb E\int_0^T\|\phi_t\sigma_t\|_{\mathcal L_2(G,\mathbb R)}^2\,\mathrm dt.
\label{eq:lambda-norm}
\end{align}
Let $\mathcal N_M:=\{\phi\in\mathcal E(H^\ast):\|\phi\|_{\Lambda_T^2(M)}=0\}$. The admissible
integrand space is the Hilbert-space completion
$$
\Lambda_T^{2}(H,\mathbb R;M)
:=
\overline{\mathcal E(H^\ast)/\mathcal N_M}^{\,\|\cdot\|_{\Lambda_T^2(M)}}.
$$
Two predictable representatives that differ only on covariance-null directions are identified.
\end{definition}

If a predictable, possibly unbounded, linear functional $\phi_t$ is defined on
$\operatorname{Ran}Q_M^{1/2}(t)$ and can be approximated in the norm
\eqref{eq:lambda-norm} by elements of $\mathcal E(H^\ast)$, we use the same symbol for the induced
element of $\Lambda_T^2(H,\mathbb R;M)$. Such a representative need not be a pointwise viable
portfolio; only its covariance-norm equivalence class is part of the Hilbert space.

\begin{proposition}[Stochastic integral isometry]\label{prop:ito-isometry-M}
For every $\phi\in\Lambda_T^{2}(H,\mathbb R;M)$ there exists a unique square-integrable martingale
$$
G_t(\phi)=\int_0^t\phi_s\,\mathrm dM_s,\qquad t\in[0,T],
$$
obtained by completion from simple predictable integrands. For
$\phi,\psi\in\Lambda_T^{2}(H,\mathbb R;M)$,
\begin{align}\label{eq:isometry-closure}
\mathbb E\bigl[G_T(\phi)G_T(\psi)\bigr]
=
\langle\phi,\psi\rangle_{\Lambda_T^2(M)}.
\end{align}
In particular, if $\phi_t(h)=\langle\xi_t,h\rangle_H$ for a predictable $H$-valued representative,
then
$$
\|\phi\|_{\Lambda_T^2(M)}^2
=
\mathbb E\int_0^T\|Q_M^{1/2}(t)\xi_t\|_H^2\,\mathrm dt.
$$
\end{proposition}

\begin{proof}
For simple predictable $H^\ast$-valued integrands the stochastic integral is the usual Hilbert-space
Itô integral applied to the Hilbert--Schmidt operator $\phi_t\sigma_t:G\to\mathbb R$. The scalar
Itô isometry gives
$$
\mathbb E\bigl[G_T(\phi)G_T(\psi)\bigr]
=
\mathbb E\int_0^T
\langle \phi_t\sigma_t,\psi_t\sigma_t\rangle_{\mathcal L_2(G,\mathbb R)}\,\mathrm dt
=
\langle\phi,\psi\rangle_{\Lambda_T^2(M)}.
$$
The quotient by $\mathcal N_M$ identifies precisely the simple integrands with zero integral
norm. Hence the integral map descends to an isometry on
$\mathcal E(H^\ast)/\mathcal N_M$ and extends uniquely to its Hilbert-space completion. If
$\phi_t(h)=\langle\xi_t,h\rangle_H$, then
$$
\|\phi_tQ_M^{1/2}(t)\|_{\mathcal L_2(H,\mathbb R)}^2
=
\|Q_M^{1/2}(t)\xi_t\|_H^2,
$$
which gives the final display.
\end{proof}

\subsubsection{Self-financing gains by closure}

A portfolio is \emph{self–financing} if no external cash is injected after the
initial date: purchases of new forwards are financed solely by selling others.
We work under the simplifying assumption of a zero short rate, so cash holdings
earn no interest. In an infinite-dimensional forward market it is most convenient
to define self-financing gains directly as the $L^2$-closure of stochastic
integrals of \emph{simple} finite-maturity integrands against
$M=\int\sigma\,\mathrm{d}W$. Fixed calendar-date buy-and-hold forwards form a regular
sub-class for which this definition agrees with the usual calendar-time wealth equation.

\paragraph{Simple strategies and their gains.}
A \emph{simple} finite-maturity strategy is a finite predictable combination of point or
window functionals. For instance, with $0\le t_1<\dots<t_K\le T$, maturity coordinates
$x_1,\dots,x_n\ge0$ and bounded $\mathcal F_{t_k}$-measurable weights $q_{k,i}$,
$$
\phi_s(\omega)=\sum_{k=1}^{K}\sum_{i=1}^{n}q_{k,i}(\omega)\,\delta_{x_i}\,\mathbf 1_{(t_k,t_{k+1}]}(s)
$$
is a rolling time-to-maturity strategy in the idealised continuum market, and its gain is
defined by
\begin{align}\label{eq:SF-musiela}
G_t(\phi)=\int_0^t \phi_s\,\sigma_s\,\mathrm{d}W_s=\int_0^t \phi_s\,\mathrm{d}M_s.
\end{align}
For a fixed calendar maturity $T_i$, the corresponding buy-and-hold strategy uses
$\delta_{T_i-s}$ instead. In the arithmetic martingale specification ($g\equiv0$), its
calendar-time gain is the elementary finite sum
$$
G_t(\phi)=\sum_{k=1}^{K}\sum_{i=1}^{n}q_{k,i}\,\bigl(F(t\wedge t_{k+1},T_i)-F(t\wedge t_k,T_i)\bigr).
$$
Using the Musiela parametrisation $f_s(x):=F(s,s+x)$ together with the identity
$\mathrm{d}F(s,T_i)=\delta_{T_i-s}\,\sigma_s\,\mathrm{d}W_s$ (which follows from the buy-and-hold
self-financing example~\ref{ex:single-forward-admissibility} below) this admits the stochastic-integral representation
\begin{align*}
G_t(\phi)=\int_0^t \phi_s\,\sigma_s\,\mathrm{d}W_s=\int_0^t \phi_s\,\mathrm{d}M_s.
\end{align*}

\paragraph{Self-financing condition by closure.}
We \emph{define} the space of self-financing gains $\mathcal G(M)$ as the
$L^2(\mathbb{P})$-closure of the linear span of all simple finite-maturity gains
$G_T(\phi)$. By the Itô isometry,
$$
\mathbb E[G_T(\phi)^2]=\|\phi\|_{\Lambda_T^{2}(M)}^{2}
=\mathbb E\int_0^T \|\phi_s\,\sigma_s\|_{\mathcal L_2(G,\mathbb R)}^{2}\,\mathrm{d}s,
$$
the map $\phi\mapsto G_T(\phi)$ extends to an isometry of $\Lambda_T^{2}(H,\mathbb R;M)$ onto
$\mathcal G(M)$. For an $H^\ast$-valued representative regular enough that $\phi_t(g_t)$ is defined,
the wealth obeys the formal equation
\begin{align}\label{eq:SF-final}
\mathrm{d}U_t=\phi_t\,g_t\,\mathrm{d}t+\phi_t\,\sigma_t\,\mathrm{d}W_t ;
\end{align}
the admissible gain space used below is, however, only the $L^2$-closure of the martingale gains
$\int\phi\,\mathrm dM$, which is all the martingale-case optimisation ($g\equiv0$) requires. The bank-account position is then reconstructed from the initial endowment and
the gains process; the optimisation below is carried out directly on stochastic gains, not on a
separately chosen decomposition $U_t=\psi_t+\phi_t(f_t)$. Under any pricing measure $\mathbb{Q}$ for which the arithmetic forwards
are martingales ($g\equiv 0$), the self-financing condition reduces to
\begin{align}\label{eq:SF-pricing}
\mathrm{d}U_t=\phi_t\,\sigma_t\,\mathrm{d}W_t,\qquad 0\le t\le T.
\end{align}

\begin{remark}[Calendar-time formal derivation]
Heuristically, for a sufficiently regular calendar-time portfolio $\varphi(t,\cdot)$ and
$f_t\in\mathcal{D}(\mathcal{A})$, the Riesz-representable wealth
$$
U_t=\langle\varphi_t,f_t\rangle_{H_w}
$$
satisfies
$$
\mathrm{d}U_t
=\langle\varphi_t,\mathrm{d}f_t\rangle_{H_w}
-\langle\varphi_t,\partial_x f_t\rangle_{H_w}\,\mathrm{d}t
$$
in Musiela coordinates $\varphi_t(x):=\varphi(t,t+x)$, $f_t(x):=F(t,t+x)$. Substituting
$$
\mathrm{d}f_t=(\partial_x f_t+g_t)\,\mathrm{d}t+\sigma_t\,\mathrm{d}W_t
$$
formally yields~\eqref{eq:SF-final}.
A rigorous justification of this formal differentiation requires both
$f_t\in\mathcal{D}(\mathcal{A})$ and $\varphi_t\in\mathcal{D}(\mathcal{A}^\ast)$ for almost every
$t$. Here, we instead define gains by the closure construction above, which avoids differentiating mild
solutions.
\end{remark}

\subsubsection*{Admissible self–financing strategies}

\begin{definition}\label{def:admissible-trading-strategy}
A predictable strategy $\Phi_t=(\psi_t,\phi_t)$ on $\mathcal{T}=[0,T]$ is \emph{admissible} if
\begin{align}\label{eq:admissible-condition}
  \phi \;\in\; \Lambda_T^{2}\bigl(H,\mathbb{R}; M\bigr),
  \qquad M_t:=\int_0^t \sigma_s\,\mathrm{d}W_s.
\end{align}
It is \emph{self–financing} under $\mathbb{Q}$ if its wealth satisfies
\begin{align}\label{eq:self-financing}
  \mathrm{d}U_t \;=\; \phi_t\,\sigma_t\,\mathrm{d}W_t,
  \qquad 0\le t\le T.
\end{align}
We write
$
\mathcal{S}_T^{2}(H,\mathbb{R};M)
$
for the collection of all admissible, self–financing strategies, and abbreviate
$\mathcal{S}_T^{2}$ when the integrator $M$ is clear from context. Here $\phi$ is an element of the
covariance-norm completion $\Lambda_T^{2}(H,\mathbb R;M)$; the pair notation $(\psi_t,\phi_t)$ denotes
a pointwise portfolio only when $\phi$ admits a predictable $H^\ast$-valued representative, and
otherwise only the gain $G_T(\phi)$ and the wealth increment~\eqref{eq:self-financing} are defined.
\end{definition}

\begin{example}[Single–maturity forward position]\label{ex:single-forward-admissibility}
Consider a buy–and–hold position in the forward with delivery at calendar date $T$. Its wealth is
$
U_t=\delta_{\,T-t}(f_t)+c,\ \ t\in[0,T],
$
with bank–account offset $c=-F_0(T)$. The forward–curve component
$\phi_t:=\delta_{\,T-t}\in H^\ast$ is deterministic and therefore predictable.

\smallskip\noindent
\emph{Admissibility.}
Recall $M=\int\sigma\,\mathrm{d}W$. For each $t$,
$Q_M(t)=\sigma_t\sigma_t^\ast\in\mathcal{L}(H)$ is the covariance operator in the sense that
$
\big[\langle M,h\rangle_H,\langle M,k\rangle_H\big]_t
=\int_0^t \langle Q_M(s)h,k\rangle_H\,\mathrm{d}s.
$
Since $\delta_{\,T-t}\in\mathcal{L}(H,\mathbb{R})$ is a bounded functional and
$Q_M^{1/2}(t)\in \mathcal{L}_2(H,H)$, their composition
$\delta_{\,T-t}\,Q_M^{1/2}(t)\in \mathcal{L}_2(H,\mathbb{R})$, with
$$
\|\delta_{\,T-t}\,Q_M^{1/2}(t)\|_{\mathcal{L}_2(H,\mathbb{R})}^2
\;\le\; \|\delta_{\,T-t}\|_{\mathcal{L}(H,\mathbb{R})}^2\,
\|Q_M^{1/2}(t)\|_{\mathcal{L}_2(H,H)}^2
\;=\; \|\delta_{\,T-t}\|^2\,\mathrm{Tr}\bigl(Q_M(t)\bigr).
$$
Assumption (A3) ensures $\mathbb{E}\big[\int_0^T \mathrm{Tr}(Q_M(t))\,\mathrm{d}t\big]
=\mathbb{E}\big[\int_0^T \|\sigma_t\|_{\mathcal{L}_2(G,H)}^2\,\mathrm{d}t\big]<\infty$, hence
$\phi\in\Lambda_T^{2}(H,\mathbb{R};M)$.

\smallskip\noindent
\emph{Self–financing.}
For a buy-and-hold forward position the calendar-time gain is, by definition,
$$
F(t,T)-F(0,T)=\int_0^t \delta_{T-s}\,\mathrm dM_s
$$
under the arithmetic martingale specification. This is one of the elementary gains used in the
closure construction above. Hence the position is self-financing in the sense of
\eqref{eq:SF-musiela}--\eqref{eq:SF-pricing}. No differentiation of a generic mild solution is
needed.
\end{example}

\section{Variance-Optimal Hedging of European Options on Forwards}
\label{sec:Variance-Optimal-Hedging}

This section studies variance-optimal hedging for European claims in the stochastic-volatility
HJMM model. It formulates the infinite-dimensional problem and characterises its solution through
the Hilbert-space Galtchouk--Kunita--Watanabe (GKW) decomposition.

In general, the resulting solution is an element of the covariance-norm completion and need not have
a pointwise bounded $H^\ast$-valued representative. This motivates finite-rank and finite-maturity
approximations: spectral projections of the full GKW integrand converge in the hedging norm, and
finite-factor terminal wealths converge whenever the associated closed gain spaces approximate the
full gain space. The
associated quadratic hedging error then admits a decomposition that separates bucket
implementation, rank truncation, and irreducible residual risk.
These results are developed in detail in Section~\ref{sec:appr-real-portf}.

\subsection{The Infinite-Dimensional Variance-Optimal Hedging Problem}
\label{sec:general-solution}

The variance–optimal hedging problem is the following infinite-dimensional projection problem.

\paragraph{Working measure.}
Throughout this section we fix an equivalent pricing (martingale) measure for the arithmetic
forward curve, that is, a measure under which the drift in~\eqref{eq:HJMM-SPDE} satisfies
$g\equiv0$ and the martingale component $M=\int\sigma\,\mathrm dW$ is a genuine square-integrable
$H$-valued martingale (cf.\ \eqref{eq:SF-pricing} and
Definition~\ref{def:admissible-trading-strategy}). To keep the notation light we denote this
pricing measure by $\mathbb P$ and write $\mathbb E$ for $\mathbb E_{\mathbb P}$; the role of the
statistical (real-world) measure is addressed in Remark~\ref{rmk:martingale-case}.
Consider the filtered probability space
$(\Omega,\mathcal{F},\{\mathcal{F}_t\}_{t\ge0},\mathbb{P})$.
Let $H := H_w$ denote the Filipovi\'c space introduced in Section~\ref{sec:forward-curve-dynamics}.
The forward curve $f_t \in H$ follows the HJMM dynamics
$$
\mathrm{d}f_t = \mathcal{A}f_t\,\mathrm{d}t + \sigma_t\,\mathrm{d}W_t,
\qquad 
\mathcal{A} = \partial_x,
$$
where $W$ is a cylindrical Brownian motion on a separable Hilbert space $G$, 
and $\sigma_t \in \mathcal{L}_2(G,H)$ is predictable. 
Define the $H$–valued, square–integrable martingale component of $f$ by
$$
M_t := \int_0^t \sigma_s\,\mathrm{d}W_s,
\qquad 
\Sigma_s := \sigma_s\sigma_s^\ast \in \mathcal{L}(H),
\qquad t \in [0,T].
$$
For all $h,k \in H$, the quadratic covariation satisfies
$$
\bigl[\langle M,h\rangle_H,\,\langle M,k\rangle_H\bigr]_t
=
\int_0^t \langle \Sigma_s h,\,k\rangle_H\,\mathrm{d}s.
$$

Let $K\in L^2(\mathbb{P})$ be an $\mathcal{F}_T$–measurable payoff representing a European claim written on the forward curve. 
Admissible trading strategies with respect to $M$ are the integrands 
$\phi \in \Lambda_T^{2}(H,\mathbb{R};M)$ from Section~\ref{sec:trading-strategies}, 
whose cumulative gain is given by
$$
G_T(\phi) = \int_0^T \phi_s\,\mathrm{d}M_s.
$$
The \emph{variance–optimal hedging problem} consists in finding $(u,\phi)$ minimising the mean–squared hedging error:
\begin{align}\label{eq:variance-optimal-problem}
\varepsilon^2
:=
\inf\Bigl\{
\mathbb{E}\Bigl[\bigl(u + \textstyle\int_0^T \phi_s\,\mathrm{d}M_s - K\bigr)^2\Bigr]
:\ 
u \in \mathbb{R},\
\phi \in \Lambda_T^{2}(H,\mathbb{R};M)
\Bigr\}.
\end{align}

\begin{remark}\label{rmk:martingale-case}
Problem~\eqref{eq:variance-optimal-problem} is posed under the pricing measure $\mathbb P$, where
the traded integrator $M$ is already a martingale. This is the \emph{martingale case} of
quadratic hedging. In this case the three classical quadratic-hedging objectives coincide: the
mean–variance optimal strategy of \cite{Schweizer_2001,CernyKallsen2007}, the locally
risk-minimising strategy, and the Galtchouk--Kunita--Watanabe projection all
return the same integrand, and the variance-optimal martingale measure is $\mathbb P$ itself. The
genuinely hard part of the general semimartingale theory, the determination of the
variance-optimal martingale measure and the associated opportunity process when the discounted
asset has a $\mathbb P$-drift, therefore does not arise here and the problem reduces to the
$L^2(\mathbb P)$ orthogonal projection developed below. We solve this projection problem in the
infinite-dimensional, operator-valued stochastic-covariance setting, which is where the analytical
difficulty lies.

Two consequences should be stated explicitly: First, the minimised error
$\varepsilon^2=\mathbb E_{\mathbb P}[N_T^2]$ is computed under the pricing measure. A risk manager
is ultimately interested in the dispersion of the hedging error under the statistical measure
$\mathbb P_{\mathrm{stat}}$; the two objectives agree when $\mathbb P_{\mathrm{stat}}=\mathbb P$,
and more generally they are comparable up to the Radon--Nikodym density
$\mathrm d\mathbb P_{\mathrm{stat}}/\mathrm d\mathbb P$ whenever this density is bounded, since then
$L^2(\mathbb P)$ and $L^2(\mathbb P_{\mathrm{stat}})$ have equivalent norms. We work under
$\mathbb P$ throughout and read our results as \emph{risk minimisation under the pricing measure};
the statistical-measure problem with a non-martingale integrator is outside our scope. Second,
because the objective is an $L^2(\mathbb P)$ projection, none of the structural results below
depends on the specific origin of $\sigma$ and only the
realised covariance density $\Sigma_t=\sigma_t\sigma_t^\ast$ enters.
\end{remark}

\subsection{GKW representation in Hilbert space}\label{sec:gkw-h}
The variance-optimal hedging problem~\eqref{eq:variance-optimal-problem} is an orthogonal
projection problem in $L^2(\Omega,\mathcal F_T,\mathbb P)$. The quotient construction in
Definition~\ref{def:Lambda2-star} removes covariance-null directions before the projection is
taken; this is essential in infinite rank, where $Q_M(t)$ is typically compact and has an
unbounded pseudo-inverse. Define the sets of attainable gains and claims:
$$
\mathcal{G}(M)
:=
\Bigl\{
\textstyle\int_0^T \phi_s\,\mathrm{d}M_s
:\ 
\phi \in \Lambda_T^{2}(H,\mathbb{R};M)
\Bigr\},
\qquad
\mathcal{C}
:=
\mathbb{R} + \mathcal{G}(M)
\subset L^2(\mathbb{P}).
$$

We use the Hilbert-space martingale representation and GKW framework of
Ouvrard~\cite{Ouvrard1975ReprsentationDM}, applied to the closed stochastic-integral subspace
generated by $M$.

\begin{align}\label{eq:isometry-hedging-norm}
\mathbb{E}\Bigl[\Bigl(\int_0^T \phi_t\,\mathrm{d}M_t\Bigr)^2\Bigr]
=
\|\phi\|_{\Lambda_T^2(M)}^2
=
\mathbb{E}\int_0^T\|\phi_t Q_M^{1/2}(t)\|_{\mathcal{L}_2(H,\mathbb{R})}^2\,\mathrm{d}t.
\end{align}

\begin{proposition}[Closedness and existence of the projection]\label{prop:closure-portfolio}
The spaces $\mathcal{G}(M)$ and $\mathcal{C}=\mathbb{R}+\mathcal{G}(M)$ are closed subspaces of
$L^2(\Omega,\mathcal{F}_T,\mathbb{P})$. Consequently, there exists a unique pair
$(\tilde{u}, \tilde{\phi}) \in \mathbb{R} \times \Lambda_T^{2}(H,\mathbb{R};M)$
attaining the minimum in \eqref{eq:variance-optimal-problem}.
\end{proposition}

\begin{proof}
By Proposition~\ref{prop:ito-isometry-M}, $\phi\mapsto G_T(\phi)$ is an isometry from the complete
space $\Lambda_T^{2}(H,\mathbb{R};M)$ into $L^2(\mathbb P)$. Hence its image
$\mathcal G(M)$ is closed. Every gain has zero mean, so
$\mathbb R\perp\mathcal G(M)$ and $\mathcal C=\mathbb R\oplus\mathcal G(M)$ is closed. The minimiser
is the orthogonal projection of $K$ onto $\mathcal C$; its scalar part is
$\tilde u=\mathbb E[K]$, and its centered part has the unique representation
$G_T(\tilde\phi)$ for some $\tilde\phi\in\Lambda_T^2(H,\mathbb R;M)$.
\end{proof}

For $Z_t:=\mathbb E[K\mid\mathcal F_t]$, the Galtchouk--Kunita--Watanabe decomposition takes
place in the covariance-norm quotient. No pointwise pseudo-inverse is used unless a separate
source condition is imposed.

\paragraph{Filtration regimes.}
The residual martingale in the theorem is defined relative to the chosen filtration. If
$\mathcal F$ is generated by the traded Brownian curve noise $W$, then under the support
condition in Corollary~\ref{cor:approx-completeness} the residual vanishes. If the filtration is
enlarged by independent volatility or covariance noise, for instance
$\mathcal F=\mathcal F^W\vee\mathcal F^B$, the residual generally contains the part of the
claim martingale driven by that non-traded noise. If the covariance driver is jump driven, as in
the BNS/OU specification of Section~\ref{subsec:affine-covariance-example}, the residual may also
contain discontinuous martingale parts orthogonal to the continuous curve martingale $M$. Thus
Theorem~\ref{thm:variance-optimal-hedging} proves orthogonality to $M$, not representation of
all $L^2$-martingales by $M$.

\begin{theorem}[Infinite-dimensional GKW decomposition and variance-optimal hedge]
\label{thm:variance-optimal-hedging}
There exist a unique $\tilde\phi\in\Lambda_T^{2}(H,\mathbb R;M)$ and a square-integrable real
martingale $N$ such that
\begin{align}\label{eq:GKWT-decomposition}
Z_t
=
Z_0+\int_0^t\tilde\phi_s\,\mathrm dM_s+N_t,
\qquad 0\le t\le T,
\end{align}
and $N$ is strongly orthogonal to $M$, i.e.
$$
[N,\langle M,h\rangle_H]\equiv0,\qquad h\in H.
$$
The variance-optimal hedge is $(\tilde u,\tilde\phi)$ with $\tilde u=\mathbb E[K]$, and
\begin{align}\label{eq:price-decomposition-final}
K
=
\tilde u+\int_0^T\tilde\phi_s\,\mathrm dM_s+N_T.
\end{align}
The minimal hedging error is
$$
\varepsilon^2=\mathbb E[N_T^2].
$$
\end{theorem}

\begin{proof}
Let $\Pi_{\mathcal C}K=\tilde u+G_T(\tilde\phi)$ be the orthogonal projection given by
Proposition~\ref{prop:closure-portfolio}. Since $\mathcal G(M)$ is centered,
$\tilde u=\mathbb E[K]$. Define
$$
N_t:=Z_t-Z_0-\int_0^t\tilde\phi_s\,\mathrm dM_s.
$$
Then $N$ is a square-integrable martingale and
$N_T=K-\Pi_{\mathcal C}K$ is orthogonal in $L^2$ to $\mathcal G(M)$. Fix $h\in H$ and let
$\theta$ be any bounded predictable real process. The terminal random variable
$$
\int_0^T \theta_s\,\langle h,\mathrm dM_s\rangle_H
$$
belongs to $\mathcal G(M)$, hence its $L^2$ inner product with $N_T$ is zero. Writing
$Y:=\langle M,h\rangle_H$ and using that $N$ and $\int\theta\,\mathrm dY$ are martingales,
integration by parts gives
$$
0=\mathbb E\Bigl[N_T\int_0^T\theta_s\,\mathrm dY_s\Bigr]
=\mathbb E\Bigl[\int_0^T\theta_s\,\mathrm d[N,Y]_s\Bigr]
\qquad\text{for all bounded predictable }\theta,
$$
so the predictable covariation $\langle N,Y\rangle$ vanishes. Since
$M=\int_0^{\cdot}\sigma_s\,\mathrm dW_s$ has continuous paths, $Y=\langle M,h\rangle_H$ is
continuous, and therefore $[N,Y]=\langle N,Y\rangle=0$ even if $N$ has jumps; that is,
$[N,\langle M,h\rangle_H]\equiv0$ for every $h\in H$ (the
Kunita--Watanabe/Ouvrard characterisation of strong orthogonality,
\cite{Ouvrard1975ReprsentationDM}). The optimality and the error formula follow
from the Pythagorean theorem in $L^2$.
\end{proof}

\begin{lemma}\label{lem:gkw-representative}
The abstract integrand $\tilde\phi$ in Theorem~\ref{thm:variance-optimal-hedging} is defined in
the covariance-norm quotient $\Lambda_T^{2}(H,\mathbb R;M)$ independently of any pointwise
representative.
Assume that the covariation between $Z$ and $M$ is represented by a predictable $H$-valued process
$q_{Z,M}$ in the sense that, for every $h\in H$,
\begin{align}\label{eq:q-representability}
\frac{\mathrm d}{\mathrm dt}[\langle M,h\rangle_H,Z]_t
=
\langle q_{Z,M}(t),h\rangle_H
\quad\text{for }(\mathrm dt\otimes\mathbb P)\text{-a.e. }(t,\omega).
\end{align}
Assume moreover that the covariation density satisfies the source condition
\begin{align}\label{eq:source-condition}
q_{Z,M}(t)\in\mathrm{Ran}\,Q_M(t)
\quad\text{and}\quad
\mathbb E\int_0^T\|Q_M^{1/2}(t)Q_M^\dagger(t)q_{Z,M}(t)\|_H^2\,\mathrm dt<\infty .
\end{align}
Assume also that the Moore--Penrose representative
$$
(t,\omega)\longmapsto Q_M^\dagger(t,\omega)q_{Z,M}(t,\omega)
$$
admits a predictable $H$-valued version. 
Then the GKW integrand has the $H$-valued representative
\begin{align}\label{eq:pseudoinverse-solution}
\tilde\xi_t=Q_M^\dagger(t)q_{Z,M}(t),\qquad
\tilde\phi_t(h)=\langle\tilde\xi_t,h\rangle_H,
\end{align}
identified modulo $\ker Q_M^{1/2}(t)$. Equivalently,
\begin{align}\label{eq:normal-equations}
Q_M(t)\tilde\xi_t=q_{Z,M}(t)
\quad\text{in }H.
\end{align}
\end{lemma}

\begin{proof}
Set $\xi_t:=Q_M^\dagger(t)q_{Z,M}(t)$, choosing the predictable version specified in the statement.
The integrability in
\eqref{eq:source-condition} says precisely that
$h\mapsto\langle\xi_t,h\rangle_H$ defines an element of the covariance-norm completion. Since the
range/source condition gives $q_{Z,M}(t)\in\mathrm{Ran}\,Q_M(t)$ and
$Q_MQ_M^\dagger=P_{\overline{\mathrm{Ran}\,Q_M}}$ (Appendix~\ref{app:operator-prelim}), $\xi_t$
solves the normal equation \eqref{eq:normal-equations},
$$
Q_M(t)\xi_t=Q_M(t)Q_M^\dagger(t)q_{Z,M}(t)
=P_{\overline{\mathrm{Ran}\,Q_M(t)}}q_{Z,M}(t)=q_{Z,M}(t).
$$
Hence, for each bounded predictable $H$-valued process $\eta$, the covariation identity
\eqref{eq:q-representability} gives
$$
\mathbb E\int_0^T\langle q_{Z,M}(t),\eta_t\rangle_H\,\mathrm dt
=
\mathbb E\int_0^T\langle Q_M(t)\xi_t,\eta_t\rangle_H\,\mathrm dt,
$$
i.e.\ the stochastic integral with representative $\xi$ has the same predictable covariations with
all scalar projections of $M$ as the GKW part of $Z$. Uniqueness in the covariance-norm quotient,
from the Itô isometry, identifies this representative with $\tilde\phi$.
\end{proof}

\begin{remark}
  \begin{enumerate}
\item[i)] Condition~\eqref{eq:q-representability} holds whenever $h\mapsto\frac{\mathrm d}{\mathrm dt}[\langle M,h\rangle_H,Z]_t$ is a bounded
functional on $H$ for a.e.\ $(t,\omega)$; a convenient sufficient condition is that $Z_t=v(t,f_t,\Sigma_t)$
for a value function $v$ that is Fr\'echet differentiable in the curve variable with
square-integrable gradient, in which case $q_{Z,M}(t)=\Sigma_t D_f v(t,f_t,\Sigma_t)$ (see
Section~\ref{subsec:affine-covariance-example}).
\item[ii)] The source condition~\eqref{eq:source-condition} has the following spectral form.
In a spectral representation $Q_M(t)e_k(t)=\lambda_k(t)e_k(t)$, the pointwise range membership
$q_{Z,M}(t)\in\mathrm{Ran}\,Q_M(t)$ is
$$
q_{Z,M}(t)\in\overline{\mathrm{Ran}\,Q_M(t)}
\quad\text{and}\quad
\sum_{\lambda_k(t)>0}
\frac{|\langle q_{Z,M}(t),e_k(t)\rangle_H|^2}{\lambda_k(t)^2}<\infty
\quad\text{for a.e.\ }(t,\omega),
$$
where support membership $q_{Z,M}(t)\in\overline{\mathrm{Ran}\,Q_M(t)}$ alone is weaker and does not
suffice; the covariance-norm admissibility
$\mathbb E\int_0^T\|Q_M^{1/2}(t)Q_M^\dagger(t)q_{Z,M}(t)\|_H^2\,\mathrm dt<\infty$ is the separate
integrability
$$
\mathbb E\int_0^T\sum_{\lambda_k(t)>0}
\frac{|\langle q_{Z,M}(t),e_k(t)\rangle_H|^2}{\lambda_k(t)}\,\mathrm dt<\infty .
$$
      \end{enumerate}
  \end{remark}

\begin{remark}
The existence of $q_{Z,M}$ in \eqref{eq:q-representability} is an additional representability
assumption. The abstract hedge in Theorem~\ref{thm:variance-optimal-hedging} exists without it.
When $Q_M(t)$ has infinite rank, $Q_M^\dagger(t)$ is usually unbounded; therefore
\eqref{eq:source-condition} is a genuine range and smoothness condition on the claim and the covariance
structure, not a consequence of the GKW theorem.
\end{remark}

\begin{lemma}\label{lem:closed-range-finite-rank}
If $T\in\mathcal{K}(H)$ is compact, then $\mathrm{Ran}(T)$ is closed if and only if $\mathrm{rank}(T)<\infty$.
\end{lemma}
\begin{proof}
Finite rank implies closed range. Conversely, suppose that $\mathrm{Ran}(T)$ is closed. If this
range were infinite-dimensional, then
$$
T:(\ker T)^\perp\to\mathrm{Ran}(T)
$$
is a bounded bijection. By the Open Mapping Theorem, its
inverse is bounded. Hence the unit ball of $\mathrm{Ran}(T)$ is the image under $T$ of a bounded
set in $(\ker T)^\perp$, and is therefore relatively compact because $T$ is compact. This is
impossible in an infinite-dimensional normed space. Thus $\mathrm{Ran}(T)$ must be finite
dimensional.
\end{proof}

Let $Q\in\mathcal{L}(H)$ be self–adjoint, non–negative, compact, with spectral resolution
$$
Q=\sum_{k\ge1}\lambda_k\,\langle\cdot,e^{(k)}\rangle_H\,e^{(k)},\qquad \lambda_k\ge0.
$$
Define the (possibly unbounded) pseudo–inverse on
$
\mathcal{D}(Q^\dagger)=\Bigl\{h\in H:\ \sum_{k:\lambda_k>0}\lambda_k^{-2}\langle h,e^{(k)}\rangle_H^2<\infty\Bigr\}
$
by
$$
Q^\dagger h=\sum_{k:\lambda_k>0}\lambda_k^{-1}\,\langle h,e^{(k)}\rangle_H\,e^{(k)}.
$$
Then $QQ^\dagger=P_{\overline{\mathrm{Ran}(Q)}}$ on $\mathcal D(Q^\dagger)$, while $Q^\dagger Q=P_{\ker(Q)^\perp}$ on all of $H$. Moreover, $Q^\dagger$ is bounded if and only if $\mathrm{Ran}(Q)$ is closed, i.e., by Lemma~\ref{lem:closed-range-finite-rank}, iff $Q$ has finite rank.

\subsection{Approximation by Realistic Portfolios and Finite–Rank Models}\label{sec:appr-real-portf}

\begin{proposition}[Viable $\varepsilon$--hedges]\label{prop:epsilon-hedge}
Let $(\tilde{u},\tilde{\phi})$ be a minimiser of \eqref{eq:variance-optimal-problem} with
residual martingale $N$ as in \eqref{eq:price-decomposition-final}.
For every $\varepsilon>0$ there exists an idealised realistic finite--maturity strategy
$\phi^{(\varepsilon)}$ in the sense of Definition~\ref{def:realistic-trading-strategy} such that
$$
\|\phi^{(\varepsilon)}-\tilde{\phi}\|_{\Lambda_T^{2}(M)}<\varepsilon
\quad\text{and}\quad
\mathbb{E}\Bigl[\bigl(G_T(\phi^{(\varepsilon)})-G_T(\tilde{\phi})\bigr)^2\Bigr]<\varepsilon^2.
$$
Consequently, the hedging error using $(\tilde{u},\phi^{(\varepsilon)})$ satisfies the exact decomposition
\begin{align}\label{eq:epsilon-hedge-bound}
\mathbb{E}\Bigl[\bigl(\tilde{u}+G_T(\phi^{(\varepsilon)})-K\bigr)^2\Bigr]
=
\mathbb{E}\Bigl[\bigl(G_T(\phi^{(\varepsilon)})-G_T(\tilde{\phi})\bigr)^2\Bigr]
+ \mathbb{E}[N_T^2]
< \varepsilon^2 + \mathbb{E}[N_T^2],
\end{align}
and can be made arbitrarily close to the optimum $\mathbb{E}[N_T^2]$.
\end{proposition}

\begin{proof}
Density of idealised realistic finite--maturity portfolios in $\Lambda_T^{2}(H,\mathbb{R};M)$ follows from
Lemma~\ref{lem:predictable-finite-maturity-density}; this yields the existence of
$\phi^{(\varepsilon)}$ with the stated $\Lambda_T^{2}(M)$-bound. The Itô isometry on $\mathcal{G}(M)$
yields the $L^2$-bound on the gains. For the error decomposition,
$\tilde{u}+G_T(\phi^{(\varepsilon)})-K=\bigl(G_T(\phi^{(\varepsilon)})-G_T(\tilde\phi)\bigr)-N_T$ by
\eqref{eq:price-decomposition-final}; since $N\perp\mathcal{G}(M)$, the cross term in the squared norm
vanishes, giving the equality in~\eqref{eq:epsilon-hedge-bound}.
\end{proof}

Spectral projections of the full infinite-rank GKW integrand converge in the hedging norm. This
does not recompute a claim or a hedge in a different finite-factor HJMM model; that distinction is
spelled out in Remark~\ref{rmk:projected-vs-restricted}.

\begin{theorem}[Spectral projection of the GKW hedge]\label{thm:finite-rank-stability}
Assume that there exist predictable finite-rank orthogonal projections $P_n(t)$ on $H$ such that
$P_n(t)Q_M(t)=Q_M(t)P_n(t)$ and $P_n(t)\uparrow P_{\overline{\mathrm{Ran}\,Q_M(t)}}$ strongly for
$(\mathrm dt\otimes\mathbb P)$-a.e.\ $(t,\omega)$. 
Let $\tilde\phi\in\Lambda_T^{2}(H,\mathbb{R};M)$ be the GKW integrand from
Theorem~\ref{thm:variance-optimal-hedging}, and define the
\emph{spectrally projected hedge} first on elementary representatives by
$$
\tilde\phi^{[n]}_t(h):=\tilde\phi_t\bigl(P_n(t)h\bigr),\qquad t\in[0,T],\ h\in H,
$$
and then by continuous extension in $\Lambda_T^2(M)$. This is the
$\Lambda_T^{2}$-orthogonal projection of $\tilde\phi$ onto the closed subspace of integrands
supported on $P_n(t)H$.
Then
\begin{align}\label{eq:projected-convergence}
  \|\tilde\phi^{[n]}-\tilde\phi\|_{\Lambda_T^{2}(H,\mathbb{R};M)}\;\longrightarrow\; 0,
  \qquad
  \mathbb{E}\Bigl[\bigl(G_T(\tilde\phi^{[n]})-G_T(\tilde\phi)\bigr)^2\Bigr]\;\longrightarrow\; 0
  \qquad \text{as } n\to\infty.
\end{align}
\end{theorem}

\begin{proof}
    Since $Q_M(t)$ is trace class, $Q_M^{1/2}(t)$
is Hilbert--Schmidt, and strong convergence of $P_n(t)$ on $\overline{\mathrm{Ran}\,Q_M(t)}$
\emph{automatically} yields the Hilbert--Schmidt tail bound
$$
\|(I-P_n(t))Q_M^{1/2}(t)\|_{\mathcal L_2(H)}\longrightarrow0
\quad\text{for }(\mathrm dt\otimes\mathbb P)\text{-a.e. }(t,\omega),
$$
by dominated convergence over the singular values; we use this consequence below.

Let $\phi$ be an elementary bounded representative. Since $P_n(t)$ commutes with
$Q_M(t)$, it also commutes with $Q_M^{1/2}(t)$ and, by assumption,
$$
\|(\phi-\phi P_n)Q_M^{1/2}\|_{\mathcal L_2(H,\mathbb R)}^2
=
\|\phi(I-P_n)Q_M^{1/2}\|_{\mathcal L_2(H,\mathbb R)}^2
\longrightarrow0
$$
for $(\mathrm dt\otimes\mathbb P)$-a.e. $(t,\omega)$. The integrand is bounded by
$\|\phi Q_M^{1/2}\|_{\mathcal L_2(H,\mathbb R)}^2$, which is integrable by
\eqref{eq:lambda-norm}; dominated convergence gives
$\|\phi-\phi P_n\|_{\Lambda_T^2(M)}\to0$ for elementary $\phi$. By density of elementary
integrands in $\Lambda_T^2(M)$ and uniform boundedness of the projections, the same holds for
$\tilde\phi$. The convergence of terminal gains follows from the isometry
\eqref{eq:isometry-hedging-norm}.
\end{proof}

\begin{remark}[Measurable spectral projections]
If $Q_M(t)$ admits a predictable eigenbasis, then the projections onto the first $n$ eigenvectors
satisfy the assumptions of Theorem~\ref{thm:finite-rank-stability}. The theorem is stated directly
in terms of predictable projections to avoid relying on a measurable ordering of eigenvectors when
eigenvalues have multiplicities or crossings. In the fixed-eigenbasis model below the projections
are deterministic, so this issue disappears.
\end{remark}

\begin{remark}[Projected hedge and the hedge in a reduced rank model]\label{rmk:projected-vs-restricted}
The integrand $\tilde\phi^{[n]}$ is the $\Lambda_T^{2}$-orthogonal projection of the
\emph{infinite-rank} hedge onto rank-$n$ supported integrands; it is \emph{not}, in general,
the GKW hedge of $K$ computed with respect to a different martingale
$M^{(n)}=\int P_n\sigma\,\mathrm{d}W$ in an approximating finite-factor HJM model.
Whenever the claim martingale $Z_t=\mathbb{E}[K\mid\mathcal{F}_t]$ is held fixed and the gain
space is restricted to spectral rank $n$, the projection $\tilde\phi^{[n]}$ coincides with the
restricted-gain-space optimiser; if the claim itself is recomputed under a different
forward-curve model, additional control on the change of optimiser is required.
\end{remark}

\paragraph{Interpretation for forward–curve options.}
For European payoffs of the form $K=h\bigl(F_T(\tau_1),\dots,F_T(\tau_m)\bigr)$ or curve functionals $K=h(f_T)$ with Lipschitz $h$, the process $q_{Z,M}(t)$ aggregates the predictable covariations between $Z$ and the maturity–wise martingale increments of $M$. 
When the representability and source conditions of Lemma~\ref{lem:gkw-representative} hold, formula
\eqref{eq:pseudoinverse-solution} gives the minimal-norm exposure $\tilde{\xi}_t$ to the continuum of
maturities that best explains, in $L^2$, the movement of $Z$ by movements of $M$. The abstract GKW
integrand exists even when this pointwise representation is unavailable.

\paragraph{Two sanity checks.}
Two extremes fix the interpretation.

\begin{example}[Orthogonal claim]
Assume $Z$ is a square–integrable martingale \emph{orthogonal} to $M$, e.g., $[Z,\langle M,h\rangle_H]\equiv 0$ for all $h\in H$ (independence is sufficient). 
Then $Z-Z_0$ is orthogonal to the gain space $\mathcal G(M)$. By
Theorem~\ref{thm:variance-optimal-hedging},
$$
\tilde{\phi}=0
\quad\text{in }\Lambda_T^2(H,\mathbb R;M),
$$
so the variance–optimal action is \emph{not to trade}. 
Financially: if the claim's pricing martingale $Z$ carries no covariation with the traded risk $M$, no hedge can reduce variance.
\end{example}

\begin{example}[Linear claim in $M$]
Suppose $Z$ is a continuous linear transform of $M$:
$$
Z_t \;=\; B(M_t)\;=\;\langle \xi, M_t\rangle_H,\qquad \xi\in H,\ B\in\mathcal{L}(H,\mathbb{R}).
$$
Then $q_{Z,M}(t)=Q_M(t)\,\xi$ and thus
$$
\tilde{\xi}_t
\;=\;
Q_M^{\dagger}(t)\,q_{Z,M}(t)
\;=\;
Q_M^{\dagger}(t)\,Q_M(t)\,\xi
\;=\;
P_{\overline{\mathrm{Ran}\,Q_M(t)}}\,\xi.
$$
Identifying integrands modulo $\ker Q_M^{1/2}(t)$, the variance–optimal functional can be chosen as
$$
\tilde{\phi}_t(h)=\langle \xi,h\rangle_H,
$$
so the optimal strategy holds the claim's linear exposure to the traded martingale part. 
\end{example}

\paragraph{Attainment, decomposition, and strict viability.}
By Proposition~\ref{prop:closure-portfolio} and Theorem~\ref{thm:variance-optimal-hedging} the
infimum in \eqref{eq:variance-optimal-problem} is attained at $(\mathbb E[K],\tilde\phi)$ with
optimal error $\mathbb E[N_T^2]$. The forward-curve positions of
Section~\ref{sec:trading-strategies} are $H^\ast$-valued, but the Moore--Penrose representative
\eqref{eq:pseudoinverse-solution} of $\tilde\phi_t$ need not be: $Q_M^{\dagger}(t)$ may be
unbounded, so $\tilde\phi$ is always a well-defined element of $\Lambda_T^{2}(H,\mathbb R;M)$ yet
may fail to be a strictly viable portfolio on the curve. This motivates a closer look at $Q_M$.

\paragraph{Structure of the covariance density.}
In our HJMM market,
$$
Q_M(t)\;=\;\Sigma_t\;=\;\sigma_t\sigma_t^{\ast}\in\mathcal{L}(H),
$$
the predictable, self–adjoint, non–negative covariance density of $M$. 
No normalisation (e.g., by a trace) is needed: for $h,k\in H$,
$$
\frac{\mathrm{d}}{\mathrm{d}t}\bigl[\langle M,h\rangle_H,\langle M,k\rangle_H\bigr]_t
\;=\;\langle \Sigma_t h,\,k\rangle_H.
$$
Since $\sigma_t\in\mathcal{L}_2(G,H)$ is Hilbert–Schmidt, $\Sigma_t$ is trace–class and hence compact.

The operator-level distinction is simple: finite-rank covariance densities have bounded
Moore--Penrose inverses, while infinite-rank compact covariance densities force claim-dependent
representability conditions.

\begin{proposition}\label{prop:dichotomy}
Let $Q_M(t)=\Sigma_t=\sigma_t\sigma_t^\ast$ be the (compact, trace-class) covariance density.
Exactly one of the following holds at each $(t,\omega)$.
\begin{enumerate}
\item[(i)] \emph{Finite rank.} $\mathrm{Ran}\,Q_M(t)$ is closed, equivalently $Q_M(t)$ has finite
rank, equivalently the pseudo-inverse $Q_M^{\dagger}(t)$ is bounded. Then the Moore--Penrose
representative \eqref{eq:pseudoinverse-solution} of $\tilde\phi_t$, when defined, is a bounded
$H^\ast$-valued (strictly viable) portfolio.
\item[(ii)] \emph{Infinite rank.} $\mathrm{Ran}\,Q_M(t)$ is dense in its closed support
$\overline{\mathrm{Ran}\,Q_M(t)}=(\ker Q_M(t))^\perp$ but is not closed in that support, and
$Q_M^{\dagger}(t)$ is unbounded. Density in all of $H$ holds only under the extra assumption that
$Q_M(t)$ is injective. 
\end{enumerate}
In both cases the optimiser $\tilde\phi$ exists in the completion
$\Lambda_T^{2}(H,\mathbb R;M)$.
\end{proposition}

\begin{proof}
By Lemma~\ref{lem:closed-range-finite-rank} a compact operator has closed range iff it has finite
rank, and $Q_M^\dagger(t)$ is bounded iff $\mathrm{Ran}\,Q_M(t)$ is closed (Appendix~\ref{app:operator-prelim}).
For a self-adjoint non-negative operator,
$\overline{\mathrm{Ran}\,Q_M(t)}=(\ker Q_M(t))^\perp$. Thus in infinite rank the range is dense
in this support but not closed there; otherwise it would be a closed infinite-dimensional range of
a compact operator. The remaining statements follow from Lemma~\ref{lem:gkw-representative},
Theorem~\ref{thm:variance-optimal-hedging}, Proposition~\ref{prop:epsilon-hedge}, and the two
examples cited in the proposition.
\end{proof}

See Example~\ref{ex:infinite-rank-nonviable} for an explicit instance in which
the variance-optimal hedge exists in the completion but has no strictly viable representative.

\begin{remark}
In infinite rank, some admissible claim martingales have formal covariation coefficients that fail
the range/source condition
$\sum_{\lambda_k(t)>0}q_k(t)^2/\lambda_k(t)^2<\infty$, and
Lemma~\ref{lem:gkw-representative} then does not provide a pointwise $H$-valued representative.
Conversely, smoother or covariance-aligned claims can have bounded representatives even in infinite
rank; the linear claim above is one example. The closed-form floor model in
Section~\ref{subsubsec:closed-form-floor} illustrates a different point: even finite-rank
covariance does not preclude a positive GKW residual when the filtration contains non-traded
covariance noise on which the claim loads. In all cases, Proposition~\ref{prop:epsilon-hedge} gives
$\varepsilon$-viable approximations within the covariance support, so boundedness alone creates no
positive viability floor.
  \end{remark}

\subsection{Finite–rank models}
Assume the volatility has an invariant finite–dimensional range:
\begin{align}\label{eq:finite-invariant-range}
\mathrm{span}\Bigl(\bigcup_{t\in[0,T]}\mathrm{Ran}(\sigma_t)\Bigr)
=\mathrm{span}\{h^{(1)},\dots,h^{(n)}\}\subset H
\quad\text{a.s.}
\end{align}
Choose $(h^{(1)},\dots,h^{(n)})$ orthonormal. Then there exist predictable
$G$-valued processes $\eta^{(1)},\dots,\eta^{(n)}$ such that
\begin{align}\label{eq:finite-rank-factorisation}
\sigma_t
=\sum_{k=1}^n h^{(k)}\otimes \eta_t^{(k)},
\qquad
M_t
=\sum_{k=1}^n h^{(k)}\,\beta_t^{(k)},
\quad
\beta_t^{(k)}:=\int_0^t \langle \eta_s^{(k)},\mathrm{d}W_s\rangle_G.
\end{align}
The factor martingales $\beta^{(k)}$ need not be independent Brownian motions. Their predictable
covariation matrix is
$$
\mathrm d[\beta^{(i)},\beta^{(j)}]_t
=
a_{ij}(t)\,\mathrm dt,
\qquad
a_{ij}(t):=\langle\eta_t^{(i)},\eta_t^{(j)}\rangle_G.
$$
Consequently,
$$
Q_M(t)=\Sigma_t=\sigma_t\sigma_t^\ast
\;=\;
\sum_{i,j=1}^n a_{ij}(t)\, h^{(i)}\otimes h^{(j)}
$$
has rank at most $n$. If $a(t)$ is diagonal, then the chosen range basis is also a covariance
eigenbasis; otherwise diagonalisation is a separate, time-dependent finite-dimensional operation.
The market is therefore $n$-factor for quadratic hedging, but completeness additionally requires
that the filtration contain no orthogonal noise beyond the factor martingales and that the factor
covariance have the appropriate rank. If those conditions fail, the residual term $N$ in
\eqref{eq:price-decomposition-final} remains non-zero.

If each $Q_M(t)$ has finite rank but no invariant finite-dimensional subspace satisfies
\eqref{eq:finite-invariant-range}, then $Q_M^\dagger(t)$ is bounded pointwise, but the active
directions may vary over time. The model is locally finite-rank but need not admit a fixed
finite-dimensional HJM realisation over $[0,T]$.

\subsection{Infinite--rank models}
Assume that $Q_M(t)=\sigma_t\sigma_t^\ast$ has infinite rank on a set of positive
$\mathrm{d}t\otimes\mathbb{P}$-measure. Then the Moore--Penrose pseudo-inverse
$Q_M^{\dagger}(t)$ is unbounded, and the variance-optimal integrand $\tilde\phi$ of
Theorem~\ref{thm:variance-optimal-hedging} exists as an element of $\Lambda_T^{2}(H,\mathbb{R};M)$
but need not be representable by a bounded ($H^\ast$-valued) integrand pointwise. Denote the
strictly viable subset
$$
\hat\Lambda:=\bigl\{\phi\in\Lambda_T^{2}(H,\mathbb{R};M): \phi \text{ admits a predictable } H^\ast\text{-valued representative}\bigr\}.
$$
By construction $\hat\Lambda\subset\Lambda_T^{2}(H,\mathbb{R};M)$ and, by density of simple
predictable finite-maturity processes in the covariance-norm completion
(Lemma~\ref{lem:predictable-finite-maturity-density}), $\hat\Lambda$ is \emph{dense} in
$\Lambda_T^{2}(H,\mathbb{R};M)$.
In particular $\tilde\phi$ admits arbitrarily good approximations in $\hat\Lambda$; the strict
viability restriction therefore does not, on its own, introduce a positive viability floor.

\begin{example}[An infinite-rank hedge with no strictly viable representative]
\label{ex:infinite-rank-nonviable}
Fix an orthonormal basis $(e_k)_{k\ge1}$ of $H$, take $G=H$, and a deterministic, time-independent
diagonal volatility $\sigma e_k=\sqrt{\lambda_k}\,e_k$ with $\lambda_k=k^{-2}$. Then
$\sum_k\lambda_k=\pi^2/6<\infty$, so \emph{(A3)} holds, $Q_M=\Sigma=\operatorname{diag}(\lambda_k)$
has infinite rank, and $Q_M^{\dagger}$ is unbounded. Write $M_t=\sum_{k\ge1}k^{-1}e_k\,\beta_t^{k}$
with independent standard Brownian motions $(\beta^{k})_{k\ge1}$, and let $\mathcal F$ be their
natural filtration. Since every $\lambda_k>0$, the operator $\sigma$ is injective with dense range,
so $\overline{\operatorname{Ran}\sigma^\ast}=G$ and the market is approximately complete with
$N\equiv0$ by Corollary~\ref{cor:approx-completeness}.

Consider the centred claim
$$
K:=\sum_{k\ge1}k^{-3/2}\,\beta_T^{k}\in L^2(\mathbb P),
\qquad \operatorname{Var}(K)=T\sum_{k\ge1}k^{-3}<\infty,
$$
which is the formal forward-curve exposure $K=\langle M_T,\eta\rangle$ to the generalised maturity
direction $\eta=\sum_k k^{-1/2}e_k\notin H$ (indeed $\|\eta\|_H^2=\sum_k k^{-1}=\infty$): the claim
loads on low-variance, high-frequency curve directions with weights decaying more slowly than the
corresponding volatilities. The variance-optimal integrand solves
$\tilde\phi_s(e_k)\sqrt{\lambda_k}=k^{-3/2}$, that is $\tilde\xi_k=k^{-1/2}$. Hence
$$
\|\tilde\xi\|_H^2=\sum_{k\ge1}k^{-1}=\infty,
\qquad\text{but}\qquad
\|\tilde\phi\|_{\Lambda_T^2(M)}^2=T\sum_{k\ge1}\lambda_k\,\tilde\xi_k^2=T\sum_{k\ge1}k^{-3}
=\operatorname{Var}(K)<\infty .
$$
Thus $\tilde\phi$ is a genuine element of $\Lambda_T^2(H,\mathbb R;M)$ with exact replication, yet it
admits no bounded $H^\ast$-valued (strictly viable) representative: this is precisely the
infinite-rank alternative of Proposition~\ref{prop:dichotomy}(ii). The spectral truncations
$\tilde\xi^{(n)}=\sum_{k\le n}k^{-1/2}e_k\in H$ are strictly viable rank-$n$ hedges with
$$
\|\tilde\phi-\tilde\phi^{[n]}\|_{\Lambda_T^2(M)}^2=T\sum_{k>n}k^{-3}\xrightarrow[n\to\infty]{}0 ,
$$
and each is approximable by finite point- or window-maturity strategies through
Lemma~\ref{lem:predictable-finite-maturity-density}. The hedge is therefore implementable to any
accuracy by viable portfolios (Proposition~\ref{prop:epsilon-hedge}) although no single finite basket
of forwards replicates the claim. This is the strict-viability failure of infinite rank in its purest
form, with $N\equiv0$ separating it cleanly from the stochastic-volatility floor of
Section~\ref{subsubsec:closed-form-floor}.
\end{example}

\paragraph{Two-component risk decomposition.}
For any implemented strategy $\phi\in\hat\Lambda$ with initial endowment $\tilde u=\mathbb{E}[K]$,
the residual quadratic risk splits exactly as
\begin{align}\label{eq:two-way-risk}
\mathbb{E}\Bigl[\bigl(\tilde u+G_T(\phi)-K\bigr)^2\Bigr]
=
\underbrace{\|\phi-\tilde\phi\|_{\Lambda_T^{2}}^2}_{\text{implementation/truncation gap}}
\;+\;
\underbrace{\mathbb{E}[N_T^2]}_{\text{unhedgeable residual}},
\end{align}
since $G_T(\phi)-G_T(\tilde\phi)\in\mathcal{G}(M)$ is orthogonal to $N_T$ in $L^2$ by GKW.
If the volatility is \emph{deterministic} and the filtration is generated by $W$, and if the
covariance support condition of \cite[Thm.~4.1]{Carmona2007} holds, then $N\equiv 0$ and the
residual is purely the implementation/truncation gap (Corollary~\ref{cor:approx-completeness}).

In stochastic-volatility settings, the part of this residual generated by covariance shocks that
are not spanned by forward trading is the \emph{stochastic-volatility floor} for claims loading on
those shocks; in the abstract theorem it is best read as the component of the claim martingale
orthogonal to the traded forward-curve gains.

\begin{corollary}[Approximate completeness under deterministic volatility]\label{cor:approx-completeness}
Suppose that:
\begin{enumerate}
\item[(a)] the volatility $\sigma$ is deterministic;
\item[(b)] the filtration $(\mathcal{F}_t)_{t\geq 0}$ coincides with the augmentation of the natural
filtration of the cylindrical Brownian motion $W$ by the $\mathbb{P}$-null sets of $\mathcal{F}$;
\end{enumerate}
Let $K\in L^2(\mathcal{F}_T,\mathbb{P})$ and write its Brownian martingale representation as
$K=\mathbb{E}[K]+\int_0^T\psi_s\,\mathrm{d}W_s$. If the fixed-claim support condition
\begin{align}\label{eq:fixed-claim-support}
\psi_t\in\overline{\mathrm{Ran}\,\sigma_t^\ast}
\quad\text{for a.e. }(t,\omega)
\end{align}
holds, then $K$ admits the representation
$K=\mathbb{E}[K]+\int_0^T\tilde\phi_s\,\mathrm{d}M_s$ with $N\equiv 0$ in
\eqref{eq:price-decomposition-final}. Consequently, for every $\varepsilon>0$ there is a realistic
$\phi^{(\varepsilon)}\in\hat\Lambda$ with
$\mathbb{E}[(\tilde u+G_T(\phi^{(\varepsilon)})-K)^2]<\varepsilon$.
In particular, if the claim-independent support condition
$$
\overline{\mathrm{Ran}\,\sigma_t^\ast}=G,
\quad\text{equivalently }\ker\sigma_t=\{0\},
\quad\text{for a.e. }t
$$
holds, then every $K\in L^2(\mathcal{F}_T,\mathbb{P})$ has this property and the market is
approximately complete.
This is the traded-Brownian-filtration regime from the paragraph preceding
Theorem~\ref{thm:variance-optimal-hedging}.
\end{corollary}

\begin{proof}
Under (a)--(b), the Brownian martingale representation gives
$K=\mathbb{E}[K]+\int_0^T\psi_s\,\mathrm{d}W_s$ for some predictable $G$-valued $\psi$.
Writing the gain integrands as $\phi_s\sigma_s=\sigma_s^\ast\xi_s$, the closed gain space
$\mathcal{G}(M)$ is the $L^2$-image of the integrands $\eta$ with
$\eta_s\in\overline{\operatorname{Ran}\sigma_s^\ast}$ for a.e.\ $s$. Hence the Brownian-representation
integrand $\psi$ of $K$ is matched by an admissible gain, and the GKW residual vanishes, precisely
when \eqref{eq:fixed-claim-support} holds; this is the support condition of
\cite[Thm.~4.1]{Carmona2007} written for the present deterministic Brownian setting. The
claim-independent dense-range condition implies \eqref{eq:fixed-claim-support} for every Brownian
representation integrand. The realistic approximation follows from Proposition~\ref{prop:epsilon-hedge}.
\end{proof}

\begin{proposition}[Three-way mean-square hedging error decomposition]\label{thm:three-way}
Let $\mathcal{C}_{\mathrm{bucket}}\subset\mathcal{C}_n\subset\mathcal{C}$ be nested closed
subspaces of $L^2(\Omega,\mathcal{F}_T,\mathbb{P})$ representing, respectively, attainable
terminal wealths from a fixed realistic finite-maturity bucket, from rank-$n$ spectrally
truncated trading, and from the full generalised market $\mathcal{C}=\mathbb{R}+\mathcal{G}(M)$.
Let $\Pi_{\mathrm{bucket}}, \Pi_n, \Pi$ denote the corresponding $L^2$-orthogonal projections.
For the hedging-norm identifications below, assume in addition that $\mathcal C_n$ is the closed
rank-$n$ spectral gain space associated with the projections in
Theorem~\ref{thm:finite-rank-stability}, and that $\mathcal C_{\mathrm{bucket}}$ is realised by a
closed finite-maturity subspace of $\mathcal C_n$.
Then the optimal bucket-and-rank hedge $\Pi_{\mathrm{bucket}}K$ satisfies
\begin{align}\label{eq:three-way-risk}
\|K-\Pi_{\mathrm{bucket}}K\|_{L^2}^{2}
=
\underbrace{\|\Pi_n K-\Pi_{\mathrm{bucket}}K\|_{L^2}^{2}}_{\text{(i) bucket implementation gap}}
\;+\;
\underbrace{\|\Pi K-\Pi_n K\|_{L^2}^{2}}_{\text{(ii) rank-}n\text{ truncation gap}}
\;+\;
\underbrace{\|K-\Pi K\|_{L^2}^{2}}_{\text{(iii) irreducible residual}}.
\end{align}
The first term equals $\|\phi^{(n)}-\phi^{(n)}_{\mathrm{bucket}}\|_{\Lambda_T^{2}}^{2}$ for the
corresponding hedges; the second equals $\|\tilde\phi-\tilde\phi^{[n]}\|_{\Lambda_T^{2}}^{2}$ and
vanishes as $n\to\infty$ by Theorem~\ref{thm:finite-rank-stability}; the third equals
$\mathbb{E}[N_T^{2}]$ and is irreducible.
\end{proposition}

\begin{proof}
This is the Pythagorean identity along a nested flag of closed subspaces; the point is that the
three projection increments have direct hedging interpretations.
By the Pythagorean theorem for nested closed Hilbert subspaces, the orthogonal
decomposition
$K - \Pi_{\mathrm{bucket}}K
= (K-\Pi K) + (\Pi K-\Pi_n K) + (\Pi_n K-\Pi_{\mathrm{bucket}}K)$
has pairwise orthogonal summands, and \eqref{eq:three-way-risk} follows by taking $L^2$-norms.
The identifications with hedging-norm quantities follow from the Itô isometry
\eqref{eq:isometry-hedging-norm} and from the additional spectral realisation of $\mathcal C_n$:
restricted to that rank-$n$ gain space, the GKW projection coincides with the spectrally projected
hedge of Theorem~\ref{thm:finite-rank-stability}.
\end{proof}

\begin{remark}[Non-optimal implemented strategies]
For an arbitrary implemented strategy $\phi_{\mathrm{bucket}}\in\hat\Lambda$ that is \emph{not}
itself the projection of $\tilde\phi$ onto $\mathcal{C}_{\mathrm{bucket}}$, only the inequality
$$
\mathbb{E}\bigl[(\tilde u+G_T(\phi_{\mathrm{bucket}})-K)^2\bigr]
\le
3\,\|\phi_{\mathrm{bucket}}-\tilde\phi^{[n]}\|_{\Lambda_T^{2}}^{2}
+3\,\|\tilde\phi^{[n]}-\tilde\phi\|_{\Lambda_T^{2}}^{2}
+3\,\mathbb{E}[N_T^{2}]
$$
follows from $(a+b+c)^2\le 3(a^2+b^2+c^2)$. The exact equality~\eqref{eq:three-way-risk}
holds only along nested orthogonal projections.
\end{remark}

\begin{remark}[Buckets as standardised delivery periods]\label{rmk:delivery-buckets}
In an energy market the bucket subspace $\mathcal C_{\mathrm{bucket}}$ is not generated by
arbitrary point maturities but by the \emph{standardised, traded delivery windows} (e.g.\ monthly,
quarterly and calendar-year base-load), i.e.\ by the averaging functionals
$\ell_{[a_i,b_i]}$ of~\eqref{eq:window-functional} for a fixed, possibly overlapping family
$\{[a_i,b_i]\}$. The bucket implementation gap (i) in~\eqref{eq:three-way-risk} is then the
$\Lambda_T^2$-distance from the rank-$n$ hedge $\tilde\phi^{[n]}$ to
$\mathrm{span}\{\ell_{[a_i,b_i]}\}$, that is, the cost of replacing a continuum of maturities by the
liquid delivery ladder. Unlike the idealised point-maturity grid, this family is finite and fixed
by the exchange, so the gap need not vanish; Proposition~\ref{prop:density} guarantees only that
\emph{refining} the delivery ladder drives it to zero.
\end{remark}

\paragraph{Constructive viable approximation.}
Assume the representability and source conditions of Lemma~\ref{lem:gkw-representative}, and let
$(P_n)_{n\ge1}$ be predictable finite-rank projections satisfying the assumptions of
Theorem~\ref{thm:finite-rank-stability}. Then the projected hedges
$\tilde\phi^{[n]}$ are the canonical finite-rank approximations in the covariance-norm quotient and
$\tilde\phi^{[n]}\to\tilde\phi$ in $\|\cdot\|_{\Lambda_T^{2}}$. If, in addition, the projections are
given by a predictable spectral selection
$Q_M(t)=\sum_{k\ge 1}\lambda_k(t)\,\langle\cdot,e^{(k)}_t\rangle_H\,e^{(k)}_t$ and the truncated
Moore--Penrose representative admits a predictable version, then the approximation has the explicit
$H^\ast$-valued representative
$$
\phi^{(n)}_t(h):=\Bigl\langle \sum_{k=1}^n \tilde{\xi}_t^{(k)}\,e^{(k)}_t,\,h\Bigr\rangle_H,
\qquad
\tilde{\xi}_t^{(k)}:=\frac{\langle q_{Z,M}(t),e^{(k)}_t\rangle_H}{\lambda_k(t)},
$$
on $\{\lambda_k(t)>0\}$, set to zero on the complementary set. Under this additional predictable
selection hypothesis, $\phi^{(n)}\in\hat\Lambda$ for each fixed $n$ and coincides with
$\tilde\phi^{[n]}$. Realistic finite-maturity portfolios in turn approximate each such
$\phi^{(n)}$ arbitrarily well by Lemma~\ref{lem:predictable-finite-maturity-density}.

\subsection{Finite-factor models as convergent restricted markets}
\label{sec:finite-factor-mosco}

Finite-factor convergence is safest when stated as stability of orthogonal projections onto
attainable wealth spaces. This avoids the discontinuity of the map
$(Q,q)\mapsto Q^\dagger q$ in infinite dimension.

\begin{proposition}[Stability under Mosco convergence of attainable wealth spaces]
\label{thm:mosco-vo-stability}

Let $\mathcal C_n$ and $\mathcal C$ be closed subspaces of
$L^2(\Omega,\mathcal F_T,\mathbb P)$. Suppose that $\mathcal C_n$ converges to $\mathcal C$ in the
Mosco sense, i.e. assume the following two conditions:
\begin{enumerate}
\item[(M1)] for every $X\in\mathcal C$, there are $X_n\in\mathcal C_n$ such that
$X_n\to X$ strongly in $L^2$;
\item[(M2)] whenever $n_j\to\infty$, $X_{n_j}\in\mathcal C_{n_j}$ and
$X_{n_j}\rightharpoonup X$ weakly in $L^2$, one has $X\in\mathcal C$.
\end{enumerate}
Then
$$
\Pi_{\mathcal C_n}Y\longrightarrow \Pi_{\mathcal C}Y
\qquad\text{in }L^2,\qquad Y\in L^2(\Omega,\mathcal F_T,\mathbb P).
$$
If $K^{(n)}\to K$ in $L^2$, then
$$
\Pi_{\mathcal C_n}K^{(n)}\longrightarrow \Pi_{\mathcal C}K
\qquad\text{in }L^2.
$$
\end{proposition}

\begin{proof}
This is the standard projection characterisation of Mosco convergence in a Hilbert space. For
completeness, fix $Y$ and set $Y_n:=\Pi_{\mathcal C_n}Y$. The sequence $(Y_n)$ is bounded because
$\|Y-Y_n\|\le\|Y\|$. Let $Y_{n_j}$ be any weakly convergent subsequence with limit $\bar Y$.
By (M2), $\bar Y\in\mathcal C$. For any $X\in\mathcal C$, choose $X_n\in\mathcal C_n$ as in (M1).
The projection inequality $\|Y-Y_n\|\le\|Y-X_n\|$ gives
$$
\limsup_j\|Y-Y_{n_j}\|\le\|Y-X\|.
$$
Minimising over $X\in\mathcal C$ yields
$\limsup_j\|Y-Y_{n_j}\|\le\|Y-\Pi_{\mathcal C}Y\|$. Weak lower semicontinuity gives the reverse
inequality with $\bar Y$ in place of $\Pi_{\mathcal C}Y$, hence $\bar Y=\Pi_{\mathcal C}Y$ and the
norms converge. Uniform convexity of $L^2$ implies strong convergence. The statement with
$K^{(n)}$ follows from
$$
\|\Pi_{\mathcal C_n}K^{(n)}-\Pi_{\mathcal C}K\|_{L^2}
\le
\|K^{(n)}-K\|_{L^2}
+\|\Pi_{\mathcal C_n}K-\Pi_{\mathcal C}K\|_{L^2}.
$$
\end{proof}

\begin{corollary}[Nested finite-factor markets]\label{cor:nested-markets}
If
$$
\mathcal C_1\subset\mathcal C_2\subset\cdots\subset\mathcal C,
\qquad
\overline{\bigcup_{n\ge1}\mathcal C_n}=\mathcal C,
$$
then $\Pi_{\mathcal C_n}K\to\Pi_{\mathcal C}K$ in $L^2$ for every $K\in L^2$.
\end{corollary}

\begin{proof}
Condition (M1) is exactly density of $\cup_n\mathcal C_n$ in $\mathcal C$. Condition (M2) follows
from $\mathcal C_n\subset\mathcal C$ and closedness of $\mathcal C$, because every weak limit of
points in $\mathcal C$ still belongs to $\mathcal C$.
\end{proof}

\begin{theorem}[Finite-factor convergence under fixed spectral truncation]
\label{thm:fixed-basis-finite-factor-convergence}
Assume that the covariance density is diagonal in a deterministic orthonormal basis $(e_k)_{k\ge1}$
of $H$:
$$
Q_M(t)e_k=\lambda_k(t)e_k,\qquad k\ge1,
$$
where $\lambda_k$ are predictable, non-negative and
$$
\mathbb E\int_0^T\sum_{k\ge1}\lambda_k(t)\,\mathrm dt<\infty.
$$
Let $P_n$ be the projection onto $\mathrm{span}\{e_1,\ldots,e_n\}$, let
$M^{(n)}:=P_nM$, and let $\mathcal C_n$ be the closed wealth space generated by integrands
supported on $P_nH$. Then $\mathcal C_n\uparrow\mathcal C$ and
$$
\overline{\bigcup_{n\ge1}\mathcal C_n}=\mathcal C.
$$
Consequently, the variance-optimal finite-factor terminal wealths converge:
$$
\Pi_{\mathcal C_n}K\longrightarrow \Pi_{\mathcal C}K
\qquad\text{in }L^2.
$$
\end{theorem}

\begin{proof}
For an elementary integrand $\phi$, define $\phi^{[n]}:=\phi P_n$. Since $P_n$ commutes with
$Q_M^{1/2}$,
$$
\|\phi-\phi^{[n]}\|_{\Lambda_T^2(M)}^2
=
\mathbb E\int_0^T\sum_{k>n}\lambda_k(t)|\phi_t(e_k)|^2\,\mathrm dt
\to0.
$$
Density of elementary integrands extends the convergence to all of $\Lambda_T^2(M)$. Moreover,
for $\phi\in\Lambda_T^2(M)$ supported on $P_nH$,
$$
\int_0^T\phi_t\,\mathrm dM_t
=
\int_0^T\phi_t\,\mathrm dM_t^{(n)}.
$$
Thus the finite-factor gain spaces are nested and dense in the full gain space. Apply
Corollary~\ref{cor:nested-markets}.
\end{proof}

\begin{corollary}\label{cor:recomputed-payoffs}
Let $(P_n)_{n\ge1}$ be the fixed spectral projections of
Theorem~\ref{thm:fixed-basis-finite-factor-convergence}. Let $f^{(n)}$ be a finite-factor mild
solution with the same initial curve and drift as $f$, and assume that
\begin{align}\label{eq:state-convergence-assumption}
\mathbb E\|f_T^{(n)}-f_T\|_H^2\longrightarrow0.
\end{align}
If $K^{(n)}:=h(f_T^{(n)})$ and $K:=h(f_T)$ with $h:H\to\mathbb R$ Lipschitz, then
$K^{(n)}\to K$ in $L^2$ and
$$
\Pi_{\mathcal C_n}K^{(n)}\longrightarrow \Pi_{\mathcal C}K
\qquad\text{in }L^2.
$$
Condition~\eqref{eq:state-convergence-assumption} holds, for example, in the exogenous truncation
case $\sigma^{(n)}_t=P_n\sigma_t$ if
$$
\mathbb E\int_0^T\|(I-P_n)\sigma_s\|_{\mathcal L_2(G,H)}^2\,\mathrm ds\to0.
$$
It also holds for the fixed-basis multiplicative equation
$\sigma^{(n)}(t,f)=P_n\sigma(t,f)$ under the Lipschitz condition
\eqref{eq:mult-lipschitz} and the tail condition
$$
\mathbb E\int_0^T\sum_{k>n}\lambda_k(t,f_t)\,\mathrm dt\to0.
$$
\end{corollary}

\begin{proof}
In the exogenous case, the stochastic convolution estimate gives
$$
\mathbb E\|f_T^{(n)}-f_T\|_H^2
\le
C_T\,\mathbb E\int_0^T\|(I-P_n)\sigma_s\|_{\mathcal L_2(G,H)}^2\,\mathrm ds
\to0.
$$
For the multiplicative case, write the mild equations for $f^{(n)}$ and $f$, use the semigroup
bound and Itô isometry, and decompose
$
P_n\sigma(s,f_s^{(n)})-\sigma(s,f_s)
=
P_n\bigl(\sigma(s,f_s^{(n)})-\sigma(s,f_s)\bigr)
-(I-P_n)\sigma(s,f_s).
$
The Lipschitz condition \eqref{eq:mult-lipschitz} gives
$$
\mathbb E\|f_t^{(n)}-f_t\|_H^2
\le
C_T\int_0^t\mathbb E\|f_s^{(n)}-f_s\|_H^2\,\mathrm ds
+C_T\,\mathbb E\int_0^t\sum_{k>n}\lambda_k(s,f_s)\,\mathrm ds .
$$
Gronwall's lemma and the tail condition imply
\eqref{eq:state-convergence-assumption}. Lipschitz continuity of $h$ yields
$K^{(n)}\to K$ in $L^2$. The projection convergence follows from
Proposition~\ref{thm:mosco-vo-stability}.
\end{proof}

\section{Two model classes}\label{sec:bns-specialization}

We now record two classes of HJMM models covered by the preceding theory. The first has exogenous
operator-valued stochastic covariance. The second has multiplicative, curve-dependent volatility
with a fixed covariance eigenbasis. The two examples separate the two roles of stochastic
volatility: random covariance creates an orthogonal residual whenever its driver is not traded, while
fixed-basis multiplicative noise gives a particularly transparent finite-factor approximation
theory.

\subsection{Affine stochastic covariance}
\label{subsec:affine-covariance-example}

Let $G=H$ and let $\Sigma$ be a predictable process with values in the cone
$\mathcal L_1^+(H)$ of positive trace-class operators. Assume
$$
\mathbb E\int_0^T\mathrm{Tr}(\Sigma_t)\,\mathrm dt<\infty.
$$
Set $\sigma_t:=\Sigma_t^{1/2}$. Then $\sigma_t\in\mathcal L_2(H)$ and
$$
\|\sigma_t\|_{\mathcal L_2(H)}^2=\mathrm{Tr}(\Sigma_t),
$$
so (A3) holds. The HJMM dynamics become
\begin{align}\label{eq:affine-covariance-hjmm}
\mathrm df_t=\mathcal A f_t\,\mathrm dt+\Sigma_t^{1/2}\,\mathrm dW_t.
\end{align}

A tractable stochastic-covariance specification is an operator-valued BNS/OU dynamics
\begin{align}\label{eq:operator-bns}
\mathrm d\Sigma_t
=
(\mathcal B\Sigma_t+\Sigma_t\mathcal B^\ast+b)\,\mathrm dt+\mathrm dL_t,
\qquad
\Sigma_0\in\mathcal L_1^+(H),
\end{align}
where $\mathcal B$ generates a positive semigroup on $\mathcal L_1(H)$, $b\in\mathcal L_1^+(H)$,
and $L$ is an $\mathcal L_1^+(H)$-valued subordinator with finite first moment. If the semigroup
generated by $\mathcal B(\cdot)+(\cdot)\mathcal B^\ast$ preserves $\mathcal L_1^+(H)$ and
$\mathbb E\,\mathrm{Tr}(L_T)<\infty$, then the mild form of \eqref{eq:operator-bns} is positive,
trace-class valued, and satisfies the integrability condition above on finite horizons. This is the
standard affine-cone setting; see, for example, the affine Hilbert-space constructions
in~\cite{CoxKarbachKhedher2022,FK24,Karbach2024HeatAffine}.

The variance-optimal hedge for a claim $K=h(f_T)$ is the GKW integrand of
Theorem~\ref{thm:variance-optimal-hedging} with
$$
Q_M(t)=\Sigma_t.
$$
If the claim martingale admits a Markovian value function
$$
v(t,f,\Sigma)=\mathbb E[h(f_T)\mid f_t=f,\Sigma_t=\Sigma],
$$
that is Fréchet differentiable in the forward-curve variable with sufficient integrability, then
Itô's formula gives the forward-curve covariation density
\begin{align}\label{eq:affine-gradient-q}
q_{Z,M}(t)=\Sigma_t\,D_fv(t,f_t,\Sigma_t),
\end{align}
and Lemma~\ref{lem:gkw-representative} yields the representative
$$
\tilde\xi_t=P_{\overline{\mathrm{Ran}\Sigma_t}}D_fv(t,f_t,\Sigma_t)
$$
whenever the source condition holds. For Fourier-transformable claims, affine transform formulae
can replace the value-gradient representation; the hedge is then obtained by differentiating the
affine transform with respect to the forward-curve state, in the spirit of the affine
variance-optimal hedging formulae of~\cite{Kallsen_Pauwels_2010}.

The stochastic-volatility floor has a transparent interpretation in this example under an explicit
noise-splitting assumption. Suppose the filtration is generated by the traded Brownian curve noise
and by the covariance driver $L$, and suppose the martingale part generated by $L$ is orthogonal to
the traded martingale $M$. Then the $L$-driven martingale component of
$v(t,f_t,\Sigma_t)$ belongs to the GKW residual $N$. Increasing the number of traded curve factors
reduces the rank truncation term, but it cannot remove covariance shocks that are not spanned by
forward trading.

More generally, let $Y$ be an $\mathbb R^m$-valued square-integrable martingale representing the
component of the non-traded covariance noise that is strongly orthogonal to $M$, continuous or
purely discontinuous. Suppose that a Markovian value process $Z_t=v(t,f_t,\Sigma_t)$ admits the
semimartingale split
$$
\mathrm dZ_t=\phi_t\,\mathrm dM_t+\gamma_t^\top\,\mathrm dY_t+\mathrm dR_t,
$$
where $R$ is strongly orthogonal to both $M$ and each component of $Y$, and
$\gamma$ is a predictable $\mathbb R^m$-valued process square-integrable with respect to the
matrix bracket measure $\mathrm d\langle Y\rangle_t$. For a Markovian model this follows, for
example, from differentiability of $v$ in the covariance state and the corresponding
square-integrability of the derivative.
Then the non-traded covariance contribution to the GKW residual satisfies
$$
\mathbb E[N_T^2]\ \ge\
\mathbb E\int_0^T \gamma_t^\top\,\mathrm d\langle Y\rangle_t\,\gamma_t,
$$
with equality for the covariance-noise part when $R\equiv0$. This term is strictly positive
exactly when $\gamma$ is nonzero on a set of positive predictable bracket measure; for
jump-driven $Y$, $\langle Y\rangle$ denotes the predictable quadratic-variation matrix,
equivalently the compensator of the jump quadratic variation.
Thus non-traded covariance noise creates a floor only for claims whose value process loads on that
noise.

The next subsection uses a diffusion-type affine covariance specialisation, not the
subordinator-driven BNS/OU dynamics in \eqref{eq:operator-bns}. The rank-one CIR choice is made only
to obtain a transparent closed-form floor.

\subsubsection{A closed-form stochastic-volatility floor}\label{subsubsec:closed-form-floor}

To show that the floor $\mathbb E[N_T^2]$ is genuinely positive, computable, and correctly scaled
by the vol-of-vol, we work out a rank-one instance in closed form. It is the simplest member of the
affine class above and isolates the floor without any rank-truncation error.

\paragraph{Model.}
Let $h\in H$ be a fixed curve direction and $e\in G$ a unit vector, and write
$\beta_t:=\langle W_t,e\rangle_G$ for the resulting scalar Brownian motion (the \emph{traded}
noise). Let $B$ be a scalar Brownian motion independent of $\beta$ (the \emph{non-traded}
covariance driver), and let the scalar variance $v$ follow the affine (CIR) dynamics
\begin{align}\label{eq:floor-cir}
\mathrm dv_t=\kappa(\theta-v_t)\,\mathrm dt+\xi\sqrt{v_t}\,\mathrm dB_t,
\qquad v_0>0,\ \kappa,\theta>0,\ \xi\ge0,
\end{align}
with $2\kappa\theta\ge\xi^2$ so that $v>0$. Take the rank-one volatility
$\sigma_t:=\sqrt{v_t}\,(h\otimes e)$, so that
$$
M_t=\int_0^t\sigma_s\,\mathrm dW_s=h\,X_t,
\qquad
X_t:=\int_0^t\sqrt{v_s}\,\mathrm d\beta_s,
\qquad
Q_M(t)=\Sigma_t=v_t\,(h\otimes h).
$$
The filtration is generated by $(\beta,B)$; the covariance driver $B$ is not spanned by the traded
gains $\mathcal G(M)=\{\int\phi\,\mathrm dX\}$. This is a diffusion-type affine specialisation of
\eqref{eq:affine-covariance-hjmm}, distinct from the subordinator-driven BNS/OU dynamics
\eqref{eq:operator-bns}, with a one-dimensional, $M$-orthogonal covariance noise.

\paragraph{Claim and hedge.}
Consider the quadratic curve claim
$$
K=\bigl(\langle M_T,\hat h\rangle_H\bigr)^2=X_T^2,
\qquad \hat h:=h/\|h\|_H^2,
$$
which is in $L^2(\mathbb P)$ because $v$ has all polynomial moments. Since
$\mathbb E[v_s\mid v_t]=\theta+(v_t-\theta)\mathrm{e}^{-\kappa(s-t)}$, the value process is affine:
$$
Z_t=\mathbb E[K\mid\mathcal F_t]=X_t^2+A(t)+\mathsf B(t)\,v_t,
\qquad
\mathsf B(t)=\frac{1-\mathrm{e}^{-\kappa(T-t)}}{\kappa},
$$
with $A$ determined by $A'(t)=-\kappa\theta\,\mathsf B(t)$, $A(T)=0$. Itô's formula gives the
martingale decomposition
$$
\mathrm dZ_t
=\underbrace{2X_t\sqrt{v_t}\,\mathrm d\beta_t}_{\text{traded ($\beta$)}}
\;+\;\underbrace{\mathsf B(t)\,\xi\sqrt{v_t}\,\mathrm dB_t}_{\text{non-traded ($B$)}},
$$
the drift terms cancelling because $A,\mathsf B$ solve the affine ODEs. The first term lies in
$\mathcal G(M)$ and the second is strongly orthogonal to $M$, so the GKW decomposition
\eqref{eq:price-decomposition-final} is read off directly: the variance-optimal hedge is
$$
\tilde\phi_t(\cdot)=2X_t\,\langle\hat h,\cdot\rangle_H
\qquad(\text{i.e. }\tilde\xi_t=2X_t\,\hat h),
\qquad
\mathrm dN_t=\mathsf B(t)\,\xi\sqrt{v_t}\,\mathrm dB_t .
$$
No source condition is needed: the residual $N$ is exhibited explicitly.

\paragraph{The floor.}
By the Itô isometry and $\mathbb E[v_t]=\theta+(v_0-\theta)\mathrm{e}^{-\kappa t}$,
\begin{align}\label{eq:floor-closed-form}
\mathbb E[N_T^2]
=\xi^2\int_0^T \mathsf B(t)^2\,\mathbb E[v_t]\,\mathrm dt
=\xi^2\int_0^T\Bigl(\frac{1-\mathrm{e}^{-\kappa(T-t)}}{\kappa}\Bigr)^2
\bigl(\theta+(v_0-\theta)\mathrm{e}^{-\kappa t}\bigr)\,\mathrm dt .
\end{align}
Three features are worth recording.
\begin{enumerate}
\item[(i)] \emph{Strict positivity.} $\mathbb E[N_T^2]>0$ whenever $\xi>0$: a non-traded covariance
shock cannot be removed by any forward-trading strategy. The market is genuinely incomplete despite
the covariance having finite (here, unit) rank, confirming that incompleteness is driven by
non-traded volatility noise rather than by the rank of $\Sigma_t$.
\item[(ii)] \emph{Vol-of-vol scaling and the deterministic limit.}
$\mathbb E[N_T^2]=\Theta(\xi^2)$, and $\mathbb E[N_T^2]\to0$ as $\xi\to0$. The $\xi=0$ case is
deterministic volatility. After reducing to the traded Brownian filtration, this is the
approximately complete regime of Corollary~\ref{cor:approx-completeness}; the floor is precisely the
obstruction created by switching on the covariance noise.
\item[(iii)] \emph{Maturity scaling.} As $\kappa\to0$ (no mean reversion, so
$\mathbb E[v_t]\to v_0$),
$\mathsf B(t)\to T-t$ and $\mathbb E[N_T^2]\to\xi^2 v_0\,T^3/3$, the familiar cubic-in-maturity
growth of an unhedged variance exposure.
\end{enumerate}
This example also shows that the truncation results are sharp in the right sense: any finite-rank
($n$-factor) hedge of $K$ reproduces the delta $\tilde\phi_t=2X_t\langle\hat h,\cdot\rangle_H$
exactly once $h\in P_nH$, so the rank-truncation gap of Proposition~\ref{thm:three-way} vanishes
from the first truncation containing $h$, while the residual term stays fixed at
\eqref{eq:floor-closed-form}. We revisit this
model numerically in Section~\ref{sec:numerics}.

\subsection{Multiplicative fixed-basis HJMM noise}
\label{subsec:multiplicative-example}

This subsection gives a tractable benchmark rather than a generic stochastic-covariance model:
the covariance eigenvectors are fixed and only the eigenvalues are stochastic or state dependent.
This restriction makes nested PCA convergence transparent.

Let $(e_k)_{k\ge1}$ be a deterministic orthonormal basis of $H$ and take $G=H$. Suppose
$$
\sigma(t,f)e_k=\sqrt{\lambda_k(t,f)}\,e_k,\qquad k\ge1,
$$
where $\lambda_k:[0,T]\times H\to\mathbb R_+$ are measurable and satisfy, for constants $C,L<\infty$,
\begin{align}
\sum_{k\ge1}\lambda_k(t,f)&\le C(1+\|f\|_H^2),\label{eq:mult-growth}\\
\sum_{k\ge1}\bigl|\sqrt{\lambda_k(t,f)}-\sqrt{\lambda_k(t,g)}\bigr|^2
&\le L^2\|f-g\|_H^2.\label{eq:mult-lipschitz}
\end{align}
Then $\sigma(t,f)\in\mathcal L_2(H)$, and the multiplicative HJMM equation
\begin{align}\label{eq:multiplicative-hjmm}
\mathrm df_t=\mathcal A f_t\,\mathrm dt+\sigma(t,f_t)\,\mathrm dW_t
\end{align}
has a unique mild solution under the standard semilinear SPDE assumptions. Its covariance density is
$$
Q_M(t)e_k=\lambda_k(t,f_t)e_k.
$$
Hence the eigenfunctions are fixed while the eigenvalues may be stochastic and state-dependent.
The finite-factor truncation
$$
\sigma^{(n)}(t,f)=P_n\sigma(t,f),\qquad
P_n=\sum_{k=1}^n e_k\otimes e_k,
$$
produces nested gain spaces. If
$$
\mathbb E\int_0^T\sum_{k>n}\lambda_k(t,f_t)\,\mathrm dt\to0,
$$
then Theorem~\ref{thm:fixed-basis-finite-factor-convergence} applies. If, in addition,
$K^{(n)}=h(f_T^{(n)})$ for a Lipschitz payoff and $f^{(n)}$ solves the projected equation, then
Corollary~\ref{cor:recomputed-payoffs} gives
$$
\Pi_{\mathcal C_n}K^{(n)}\to\Pi_{\mathcal C}K
\quad\text{in }L^2.
$$

A simple parametrisation is
$$
\lambda_k(t,f)=a_k\,\ell_k(\langle f,r_k\rangle_H),
\qquad
\sum_{k\ge1}a_k<\infty,
$$
where $\ell_k$ are bounded positive Lipschitz functions and $r_k\in H$ are uniformly bounded. If
the Lipschitz constants of the functions $\ell_k^{1/2}$ are uniformly bounded after multiplication
by $a_k^{1/2}$, then \eqref{eq:mult-growth}--\eqref{eq:mult-lipschitz} hold. This captures
curve-dependent volatility levels while preserving the fixed PCA directions needed for a clean
finite-factor convergence theorem.

\begin{remark}[Lifted Volterra equations]
Markovian lifts of stochastic Volterra equations provide another source of HJM-type SPDEs. For
example, a Volterra equation
$$
X_t=g_0(t)+\int_0^t K(t-s)b(X_s)\,\mathrm ds
       +\int_0^t K(t-s)c(X_s)\,\mathrm dB_s
$$
can often be represented by an infinite-dimensional Markov process $Y_t$ whose evolution is a
transport-type SPDE and for which $X_t$ is recovered by a continuous linear functional; see
\cite{CT20}. The quadratic hedging theory above applies to the lifted martingale part once the
traded gain space is specified. If the lifted curve is itself a forward curve, the present maturity
strip interpretation applies directly. If the lift is a latent volatility state, as in rough
volatility models, then the GKW projection should be taken with respect to the actually traded price
martingales, and the untraded lift contributes to the orthogonal residual. This is a natural
extension but not needed for the main forward-curve results of this paper.
\end{remark}

\section{Finite-factor implementation workflow}\label{sec:numerics}

The theorems above suggest the following implementation procedure.
\begin{enumerate}[label=\textbf{Step \arabic*.},leftmargin=*]
\item Specify or estimate the full covariance kernel and verify the square-integrability
assumptions.
\item Choose a spectral dimension $n$. In the fixed-eigenbasis case this is the PCA subspace
$P_nH=\mathrm{span}\{e_1,\ldots,e_n\}$; in empirical applications the basis may be estimated
from a regularised covariance kernel.
\item Compute the restricted variance-optimal projection $\Pi_{\mathcal C_n}K$ or its integrand
representative.
\item Project the rank-$n$ integrand onto the liquid bucket instruments
$\{\ell_{[a_i,b_i]}\}$ using Proposition~\ref{prop:epsilon-hedge}.
\item Report the bucket, rank, and residual terms in the decomposition
\eqref{eq:three-way-risk}.
\end{enumerate}

The error identity in Proposition~\ref{thm:three-way} separates the three losses that arise in this
workflow. The bucket term measures the cost of replacing a rank-$n$ continuum exposure by a finite
set of traded delivery instruments (Remark~\ref{rmk:delivery-buckets}). The rank term measures the
cost of discarding covariance directions beyond $P_nH$ and vanishes under
Theorem~\ref{thm:finite-rank-stability} or Theorem~\ref{thm:fixed-basis-finite-factor-convergence}.
The residual term $\mathbb E[N_T^2]$ is structural: it is the component of the claim martingale
generated by risk sources that are not spanned by forward trading.

\subsection{A reproducible synthetic study}\label{subsec:synthetic-study}

We illustrate the workflow on the closed-form model of
Section~\ref{subsubsec:closed-form-floor}, generalised to many curve factors so that all three
error terms are simultaneously active. The discretised benchmark covariance is the retained
$N=40$ Karhunen--Lo\`eve approximation of the exponential (Mat\'ern-$\tfrac12$) kernel
$k(x,y)=s^2\mathrm{e}^{-|x-y|/\rho}$ on a time-to-maturity horizon
$[0,\Theta_{\max}]$, with eigenpairs $(a_k,e_k)$. The whole covariance is modulated by a single
CIR variance $V_t$ \eqref{eq:floor-cir} driven by a non-traded Brownian motion.

The synthetic experiment works in driftless gain coordinates. The Musiela shift is suppressed, and
the claim is written on the traded martingale part
$$
K=G_T^2,\qquad G_t:=\langle M_t,w\rangle_H,
$$
where $w$ is the representer of the base-load delivery window. This is the many-factor analogue
of Section~\ref{subsubsec:closed-form-floor}; it is not a state claim $h(f_T)$. Therefore
$G$ is a martingale and the GKW integrand is exactly
$\tilde\phi_t=2G_t\,\langle w,\cdot\rangle_H$, with no adjoint Musiela semigroup term. In this
finite-dimensional numerical benchmark, $H$ is the $L^2$ maturity-grid space induced by the
quadrature weights, so $w$ is the corresponding $L^2$-representer of the delivery-window
average rather than the Filipovi\'c-space representer from Section~\ref{sec:trading-strategies}.
Thus the benchmark tests the projection geometry of the martingale gain space rather than a full
Musiela-shifted SPDE simulation. In particular, it validates the projection geometry and the floor,
not the state-recomputation of Corollary~\ref{cor:recomputed-payoffs}, since the claim $K=G_T^2$ is
written directly on the martingale part.

Up to the single positive constant $R:=\mathbb E\int_0^T G_t^2 V_t\,\mathrm dt$, the three terms
of Proposition~\ref{thm:three-way} are
$$
\text{(ii) rank-$n$ gap}=4R\sum_{k>n}a_k w_k^2,
\qquad
\text{(iii) floor}=\mathbb E[N_T^2],
$$
$$
\text{(i) bucket gap}=4R\,\mathrm{dist}_A\bigl(P_nw,\ \mathrm{span}\{P_n g_i\}\bigr)^2,
$$
where $w_k=\langle w,e_k\rangle_H$, the $g_i$ are the representers of a delivery ladder
$\{\ell_{[a_i,b_i]}\}$, and $\mathrm{dist}_A$ is the distance in the covariance inner product
$\langle u,v\rangle_A=\sum_k a_k u_kv_k$ induced by the hedging norm. The constant $R$ and the floor
are estimated by a Monte-Carlo simulation of $(V,G,N)$; the remaining geometry is exact given
$R$, whose own Monte-Carlo standard error is well under $1\%$ and
is dominated by the floor cross-check below.

\begin{center}
\begin{tabular}{ll}
\hline
Quantity & Value / method\\
\hline
CIR variance & $\kappa=2.0,\ \theta=0.04,\ \xi=0.15,\ V_0=0.04,\ T=1.0$\\
Feller condition & $2\kappa\theta=0.16\ge \xi^2=0.0225$\\
Kernel & $k(x,y)=s^2\mathrm{e}^{-|x-y|/\rho}$, $s=1.0,\ \rho=0.25$\\
Maturity grid & $\Theta_{\max}=3.0$, $N_x=600$, retained KL factors $N=40$\\
Claim window & base-load window $[0.5,1.5]$\\
Delivery ladders & nested dyadic $m\in\{2,4,8,16,32\}$ at fixed rank $N^\ast=8$ (Fig.~\ref{fig:study}, right)\\
Monte Carlo & $200{,}000$ paths, $250$ steps\\
Variance scheme & full-truncation Euler for the CIR process\\
Estimators & $R$ by path average of $\int_0^T G_t^2V_t\,\mathrm dt$;
floor by $N_T^2$\\
\hline
\end{tabular}
\end{center}

The many-factor floor benchmark used for the cross-check is
$$
\mathbb E[N_T^2]
=c_w^2\xi^2\int_0^T\mathsf B(t)^2\,\mathbb E[V_t]\,\mathrm dt,
\qquad
c_w=\sum_k a_k w_k^2,\qquad
\mathsf B(t)=\frac{1-\mathrm{e}^{-\kappa(T-t)}}{\kappa}.
$$
With the parameters above, the analytic benchmark, evaluated in the code by fine quadrature, is
$\floorCF$, the Monte-Carlo estimate is $\floorMC$, and the estimated standard error is
$\floorSE$.

\begin{figure}[t]
\centering
\includegraphics[width=0.49\linewidth]{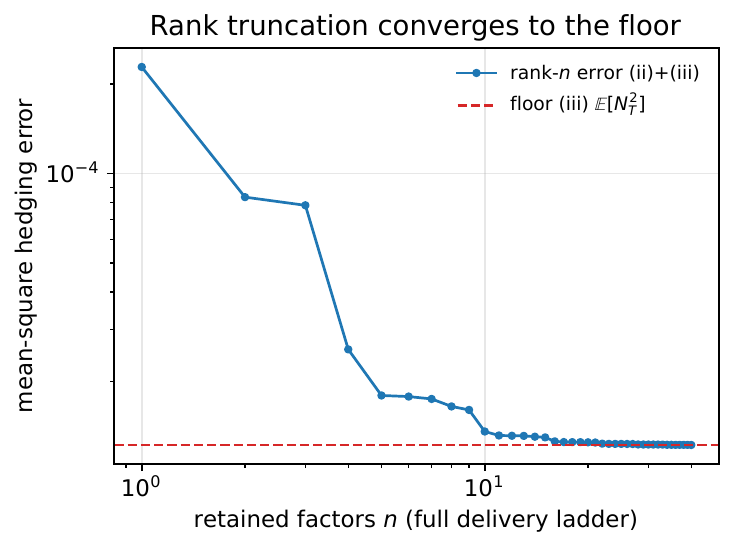}\hfill
\includegraphics[width=0.49\linewidth]{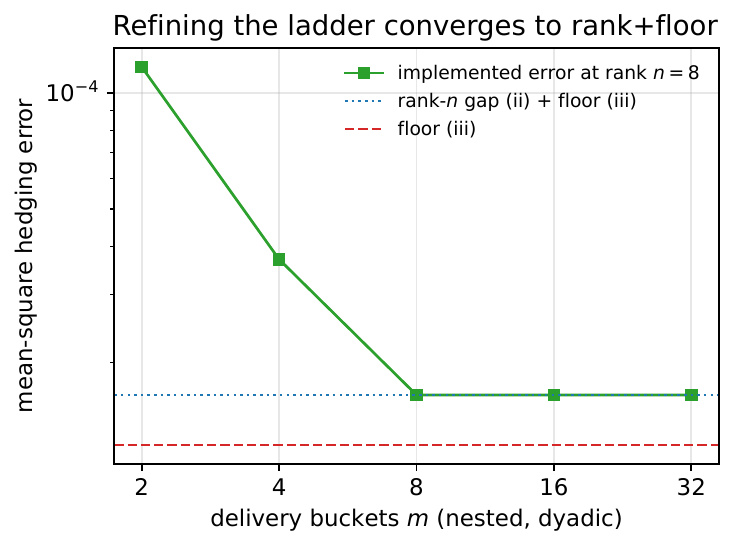}
\caption{The two panels separate the two convergence questions of
Proposition~\ref{thm:three-way}. \emph{Left (rank convergence):} with the full continuum of
maturities available, the rank-only error $(\mathrm{ii})+(\mathrm{iii})$, rank gap plus floor, decays
monotonically as covariance factors are retained and converges down to the stochastic-volatility
floor (iii) $\mathbb E[N_T^2]$. \emph{Right (implementation convergence):} at a fixed spectral rank
$N^\ast=8$, refining a nested dyadic delivery ladder $m\in\{2,4,8,16,32\}$ drives the implemented
error $(\mathrm i)+(\mathrm{ii})+(\mathrm{iii})$ monotonically down to its limit, the residual rank
gap plus the floor (dotted), which still lies strictly above the floor (dashed). The nesting
$\mathrm{span}\{P_{N^\ast}g_i^{(m)}\}\subset\mathrm{span}\{P_{N^\ast}g_i^{(2m)}\}$ of dyadic ladders
makes both curves monotone.}
\label{fig:study}
\end{figure}

Figure~\ref{fig:study} reports the outcome. In the left panel the rank-only error falls by more than
an order of magnitude as the leading covariance factors are retained and then flattens onto the
floor. This illustrates the spectral projection geometry behind
Theorem~\ref{thm:fixed-basis-finite-factor-convergence}: adding factors removes the rank gap (ii) but
never the floor (iii). The right panel fixes the rank at $N^\ast=8$ and instead refines the traded
delivery ladder; the implemented error decreases monotonically to the residual rank gap plus the
floor, illustrating Remark~\ref{rmk:delivery-buckets}, refining the ladder closes the bucket gap
(i), while the rank gap at $N^\ast$ and the floor remain. As a correctness check on the whole
pipeline, the Monte-Carlo floor $\floorMC$ matches the closed-form value $\floorCF$ of
\eqref{eq:floor-closed-form} to within $\floorRelErr$ (about one Monte-Carlo standard error), so the
simulated residual is indeed the analytic stochastic-volatility floor.

\paragraph{Explained variance is the wrong truncation criterion.}
Retaining the $n$ factors of largest variance $a_k$, the principal-component criterion, need not
minimise the hedging error, which by~\eqref{eq:three-way-risk} is governed by the covariance-weighted
exposure $a_kw_k^2$. Figure~\ref{fig:pca} illustrates this for a claim whose weight, besides the
base-load window, also places mass on a single low-variance factor $e_{k_0}$ ($k_0=13$): we compare
the covariance-weighted truncation gap $\sum_{k\notin S}a_kw_k^2$ under principal-component selection
$S=\{n\text{ largest }a_k\}$ against the hedging-norm-optimal selection
$S_\star(n)=\arg\max_{|S|=n}\sum_{k\in S}a_kw_k^2$. The variance criterion omits the low-variance but
hedging-relevant factor until $n\ge k_0$ and carries a gap larger by up to a factor of several
hundred in between, whereas the hedging norm selects $e_{k_0}$ already at $n=2$. Rank reduction must
therefore be assessed in the covariance hedging norm rather than by explained variance, consistent
with~\cite{Con05}.

\begin{figure}[t]
\centering
\includegraphics[width=0.62\linewidth]{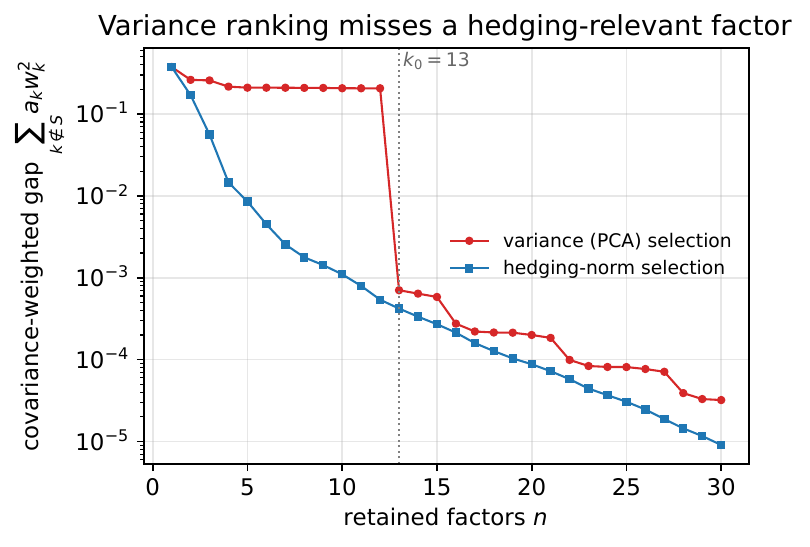}
\caption{Covariance-weighted truncation gap $\sum_{k\notin S}a_kw_k^2$ under variance
(principal-component) selection versus the hedging-norm-optimal selection, for a claim that also
loads on a low-variance factor $e_{k_0}$ ($k_0=13$, dotted line). Ranking factors by explained
variance leaves the hedging-relevant factor out until $n\ge k_0$; the hedging norm selects it
immediately. The plotted quantity is purely spectral (no Monte-Carlo).}
\label{fig:pca}
\end{figure}

The study is synthetic by design: it isolates the three error terms in a discretised benchmark where
the truth is known. An empirical study would add covariance-kernel calibration, SPDE
discretisation, and regression for the claim value gradient; the convergence results identify which
finite-dimensional objects such a study should approximate and which error terms it should report.

\section{Conclusion}\label{sec:conclusion}
This paper develops a variance-optimal hedging framework for European claims written on forward curves in
infinite-dimensional HJMM models. The construction treats the martingale part of the curve as the traded risk and defines admissible gains through the covariance-norm completion of
finite-maturity strategies. This yields a Galtchouk--Kunita--Watanabe decomposition in the correct
Hilbert-space quotient and identifies the variance-optimal hedge as an orthogonal projection in
$L^2$.

The approximation theory has two parts. Spectral projections of the full GKW hedge converge in the
covariance hedging norm. More generally, finite-factor variance-optimal terminal wealths converge
whenever the associated attainable wealth spaces converge to the full wealth space in the Mosco
sense; nested PCA truncations with a fixed covariance eigenbasis provide a transparent sufficient
condition. Along nested bucket and rank spaces the mean-square hedging error decomposes into a
bucket implementation term, a finite-rank truncation term, and an irreducible residual.

The examples show that the theory covers both exogenous operator-valued stochastic covariance and
state-dependent multiplicative HJMM noise. In the affine stochastic-covariance case, untraded
covariance shocks naturally contribute to the residual martingale. In the multiplicative
fixed-eigenbasis case, state-dependent eigenvalues preserve the nested PCA geometry needed for
finite-factor convergence.

The main limitations are the idealisation of frictionless trading, the use of square-integrable
claims, and the assumption that forward-curve gains can be approximated by liquid maturity buckets.
Extensions to transaction costs, jump-driven forward curves, multi-curve interest-rate models, and
lifted Volterra or rough-volatility dynamics are natural directions for further work.

\appendix

\section{Operator-theoretic preliminaries}\label{app:operator-prelim}

Let $G$ and $H$ be real separable Hilbert spaces. We write $\mathcal{L}(G,H)$ for bounded
linear operators and use the Hilbert--Schmidt and trace-class ideals
$$
\mathcal{L}_2(G,H)\subset \mathcal{K}(G,H),
\qquad
\mathcal{L}_1(G,H)\subset \mathcal{L}_2(G,H),
$$
where $\mathcal{K}$ denotes compact operators. For $B\in\mathcal{L}_2(G,H)$,
$$
\|B\|_{\mathcal{L}_2(G,H)}^{2}
=
\sum_{n\ge 1}\|B g_n\|_H^{2}
\quad\text{for any orthonormal basis }(g_n)\text{ of }G,
$$
and the norm is independent of the basis. The following composition bounds hold:
if $A\in\mathcal{L}(H)$, $B\in\mathcal{L}_2(G,H)$, $C\in\mathcal{L}(G)$, then
$$
\|A B C\|_{\mathcal{L}_2(G,H)}\le \|A\|_{\mathcal{L}(H)}\,\|B\|_{\mathcal{L}_2(G,H)}\,\|C\|_{\mathcal{L}(G)};
$$
if $B,C\in\mathcal{L}_2(H)$, then $CB\in\mathcal{L}_1(H)$ and
$$
\|CB\|_{\mathcal{L}_1(H)}\le \|C\|_{\mathcal{L}_2(H)}\,\|B\|_{\mathcal{L}_2(H)}.
$$
For self-adjoint, non-negative compact $Q\in\mathcal{L}(H)$ with spectral resolution
$Q=\sum_{k\ge 1}\lambda_k\,e^{(k)}\otimes e^{(k)}$, the Moore--Penrose pseudo-inverse is
the densely-defined operator
$$
\mathcal{D}(Q^{\dagger})
=
\Bigl\{h\in H:\ \sum_{k:\lambda_k>0}\lambda_k^{-2}|\langle h,e^{(k)}\rangle_H|^{2}<\infty\Bigr\},
\qquad
Q^{\dagger}h
=
\sum_{k:\lambda_k>0}\lambda_k^{-1}\langle h,e^{(k)}\rangle_H\,e^{(k)}.
$$
It satisfies $QQ^{\dagger}=P_{\overline{\mathrm{Ran}(Q)}}$ on $\mathcal D(Q^{\dagger})$, while
$Q^{\dagger}Q=P_{\ker(Q)^\perp}$ holds on all of $H$ (since $Qh\in\mathrm{Ran}(Q)\subseteq\mathcal
D(Q^{\dagger})$ for every $h$). The operator $Q^{\dagger}$ is bounded if and only if
$\mathrm{Ran}(Q)$ is closed; by Lemma~\ref{lem:closed-range-finite-rank}, this is the case if and
only if $Q$ has finite rank.
In infinite rank the range is not closed, so $P_{\overline{\mathrm{Ran}(Q)}}$ must not be
replaced by a projection onto $\mathrm{Ran}(Q)$.
A standard reference for these facts is~\cite{DaPratoZabczyk2014}; see also~\cite{PeszatZabczyk2007}.

\clearpage{}

\end{document}